\documentclass[11pt]{article}
\usepackage[utf8]{inputenc}
\usepackage{style}
\usepackage{framed}
\usepackage[parfill]{parskip}
\usepackage{setspace}
\usepackage{hyperref}

\allowdisplaybreaks

\usepackage[utf8]{inputenc}
\usepackage{dsfont}
\usepackage{framed}
\usepackage{xspace}

\newcommand{\NP}{{{\mathrm{NP}}}}
\newcommand{\FPT}{{{\mathrm{FPT}}}}
\newcommand{\W}{{{\mathrm{W}}}}
\newcommand{\N}{{{\mathbb{N}}}}
\newcommand{\R}{\mathbb{R}}
\newcommand{\ind}{\mathds{1}}
\newcommand{\leqs}{\leqslant}
\newcommand{\geqs}{\geqslant}

\DeclareMathOperator{\cost}{cost}
\DeclareMathOperator*{\opt}{OPT}

\DeclareMathOperator*{\Ex}{\mathbb{E}}
\newcommand{\shf}{\text{shift}}
\newcommand{\fr}{\text{frac}}
\newcommand{\bad}{\text{BAD}}
\newcommand{\vccubic}{\textsc{Cubic Vertex Cover}\xspace}
\newcommand{\poly}{\text{poly}}
\renewcommand{\leq}{\leqs}
\renewcommand{\geq}{\geqs}

\newcommand{\cS}{\mathcal{S}}
\newcommand{\cR}{\mathcal{R}}
\newcommand{\bs}{\mathbf{s}}
\newcommand{\bw}{\mathbf{w}}
\newcommand{\Vi}{V_{\text{int}}}

\newcommand{\scr}{\mathrm{sc}}
\newcommand{\scrdiff}{\ensuremath{\mathit{sum\hbox{-}diff}}}
\newcommand{\diffmax}{\ensuremath{\mathit{max\hbox{-}diff}}}

\newcommand{\hso}{\text{half-seq1}}
\newcommand{\hst}{\text{half-seq2}}
\newcommand{\copeland}{\text{Copeland}^{\alpha}}

\newtheorem{theorem}{Theorem}
\newtheorem{lemma}[theorem]{Lemma}
\newtheorem{definition}[theorem]{Definition}

\newtheorem{corollary}[theorem]{Corollary}

\newtheorem{proposition}[theorem]{Proposition}

\newtheorem{claim}[theorem]{Claim}

\title{Approximation and Hardness of Shift-Bribery
  \thanks{An extended abstract of this work appears in AAAI'19.}}

\author{
Piotr Faliszewski\thanks{
AGH University, Krakow, Poland.
Email: faliszew@agh.edu.pl}
\and
Pasin Manurangsi\thanks{
UC Berkeley, USA.
Email: pasin@berkeley.edu}
\and
Krzysztof Sornat\thanks{
University of Wrocław, Poland.
Email: krzysztof.sornat@cs.uni.wroc.pl}
}

\usepackage{nicefrac}

\begin{document}

\maketitle

\begin{abstract}
  In the \textsc{Shift-Bribery} problem we are given an election, a
  preferred candidate, and the costs of shifting this preferred
  candidate up the voters' preference orders. The goal is to find such
  a set of shifts that ensures that the preferred candidate wins the
  election.  We give the first polynomial-time approximation scheme
  for the \textsc{Shift-Bribery} problem for the case of positional
  scoring rules, and for the Copeland rule we show strong
  inapproximability results.
\end{abstract}

\section{Introduction}

We provide approximation algorithms and
inapproximability results for the \textsc{Shift-Bribery} problem,
introduced by Elkind et al.~\cite{EFS09} to capture
the idea of campaigning in elections.  Briefly put, we are given an
election where each voter ranks the candidates from the most to the
least appealing one, and our goal is to ensure that a given preferred
candidate becomes the winner. To this end, we can shift this candidate
up within the voters' preference orders, but each such shift comes
with a price (which, for example, measures the difficulty of
convincing the voter that our candidate is more appealing than the
voter originally thought).  Naturally, we are interested in finding as
cheap a solution as
possible.

While the \textsc{Shift-Bribery} problem was introduced in the context
of buying votes and in the context of campaigning, bribery problems
have a number of other applications (see, e.g., the works on the
margin of victory problem~\cite{MRS11,Car11,Xia12} and on measuring
candidate success~\cite{FST17}; see also the original paper of
Faliszewski et al.~\cite{FaliszewskiHH09} and the survey of
Faliszewski and Rothe~\cite{FR16}). For example, a Formula~1 season
consists of about 20 races, where each race can be seen as a voter
ranking the candidates (the drivers) in the order in which they
finished the race. For each finishing position, there is an associated
number of points and the driver who collects most points becomes the
world champion (i.e., this ``election'' uses a \emph{positional
  scoring rule} as a voting rule). We can use the
\textsc{Shift-Bribery} problem to measure how close each driver was to
winning the world championship. For example, we can set the price for
shifting a driver up by some $t$ positions in a given race to be the
difference between the finishing times of the driver and whoever
ranked $t$ positions higher. Then, the cheapest shift bribery
corresponds to the smallest speed-up that the driver needed to become
the world champion. As argued by Faliszewski et al.~\cite{FST17}, such
values can be far more informative than the score differences between
the drivers.  Bribery problems also appear in the contexts of
lobbying~\cite{CFRS07,BEFGMR14}, rating systems~\cite{GST18}, or in
combinatorial domains~\cite{BEER15}.

With the exception of a few simple voting rules, such as the
$k$-Approval family of rules and the Bucklin rule,
\textsc{Shift-Bribery} tends to be $\NP$-hard (see the works of Elkind
et al.~\cite{EFS09} and Schlotter et al.~\cite{SFE17}). Indeed, this
is the case, e.g., for Borda, Copeland, Maximin~\cite{EFS09} and
various elimination-based rules~\cite{MNRS18}.  Yet, in many cases it
can be solved quite effectively.  For example, for the case of Borda
there is a polynomial-time 2-approximation algorithm of Elkind et
al.~\cite{EFS09,EF10} and several $\FPT$ algorithms of Bredereck et
al.~\cite{BCFNN16}. On the other hand, for the case of Copeland,
\textsc{Shift-Bribery} is $\W[1]$-hard for many natural
parameters~\cite{BCFNN16}\footnote{One notable exception is the
  parameterization by the number of candidates~\cite{KMM17}.} and the
best known polynomial-time approximation algorithm has linear
approximation ratio~\cite{EF10}.

In fact, the difference between the Borda rule (and, in general, the
positional scoring rules) and the Copeland rule is even more
striking. We show that the former can be solved nearly perfectly in
polynomial time, whereas for the latter we give strong
inapproximability results:

\begin{enumerate}
\item Our main contribution is the first polynomial-time approximation
  scheme (PTAS) for \textsc{Shift-Bribery} for positional scoring
  rules (Theorem~\ref{thm:ptas-borda-main}).
  In fact, our algorithm works even for the case where the
  scoring vectors are different for different voters. Our algorithm
  uses linear programming and, in particular, basic solutions of
  linear programs. For the case of unit prices (i.e., for the case
  where each unit shift has the same cost) we even obtain an EPTAS,
  i.e., a PTAS for which the non-polynomial factors in the running
  time depend on the approximation ratio
  only (Theorem~\ref{thm:eptas-borda-unit}).
  We also show a simple combinatorial PTAS for this case
  (Theorem~\ref{thm:simple-ptas-borda}).

\item For the case of the Copeland rule, we give a reduction that
  preserves approximation ratios up to some polynomial from the
  \textsc{Densest $k$-Subgraph} (\textsc{D$k$S}) problem to
  \textsc{Shift-Bribery} (Theorem~\ref{thm:copeland}).
  Since it is generally believed that Densest
  $k$-Subgraph is hard to approximate up to a polynomial
  factor~\cite{BCVGZ12,Man17}, the same beliefs transfer to the case
  of Copeland-\textsc{Shift-Bribery}. In particular, this gives an 
  almost-polynomial ratio hardness of approximating 
  Copeland-\textsc{Shift-Bribery} under the ETH and Gap-ETH assumptions
  (Corollary~\ref{cor:copeland}).
  We also show that under Gap-ETH, Copeland-\textsc{Shift-Bribery}
  does not admit an $\FPT$ approximation scheme for the
  parameterization by the number of unit shifts
  (Theorem~\ref{thm:copeland-unit-fpt}) or by the number of affected
  voters (Theorem~\ref{thm:copeland-unit-fpt-voter}), even for the
  case of unit prices. This is in contrast to parameterization by the
  number of voters or by the number of candidates, for which $\FPT$
  approximation schemes are known to exist~\cite{BCFNN16}.

\end{enumerate}

Together with the results of Elkind et al.~\cite{EFS09,EF10}
and Bredereck et
al.~\cite{BCFNN16}, our work gives a nearly complete view of the
complexity and approximability of \textsc{Shift-Bribery} for
positional scoring rules and the Copeland rule.

\section{Preliminaries}

For each positive integer $r \in \N$, we write $[r]$ to denote the
set $\{1, \dots, r\}$, and by $[0]$ we mean  the empty set. For an
event $X$, we write $\ind[X]$ to denote the indicator 
function 
such that $\ind[X] = 1$ if $X$ occurs and
$\ind[X] = 0$ otherwise.

\paragraph{Elections.}
An \emph{election} $E = (C, V, \{\succ_v\}_{v \in V})$ consists of a
set $C$ of $m$ candidates, a set $V$ of $n$ voters, and the collection
$\{\succ_v\}_{v \in V}$ of the voters' preference orders.  For each
voter $v$, preference order $\succ_v$ gives $v$'s ranking of the
candidates from the most to the least desirable one. 
For a preference order $\succ_v$, we write 
$\pi_v \colon [m] \to C$ to denote a function such that
$\pi_v(1) \succ_v \pi_v(2) \succ_v \cdots \succ_v \pi_v(m)$ (in other
words, $\pi_v(i)$ is the candidate that $v$ ranks on the $i$-th
position). Depending on the context,
we either specify voters' preference orders directly or
via the $\pi$ functions.

Given two candidates $c, c' \in C$, we write $V_{c \succ c'}$ to
denote the set of all voters $v \in V$ that prefer $c$ over~$c'$.

\paragraph{Voting Rules.}
A \emph{voting rule} $\cR$ is a function that for each election
$E = (C, V, \{\succ_v\}_{v \in V})$ outputs the set
$\cR(E) \subseteq C$ of this election's tied winners.  We focus on the
class of positional scoring rules and on the Copeland rule.

Consider a setting with $m$ candidates.  Under a \emph{positional
scoring rule} $\cR_\bw$, we have a vector
$\bw = (w_1, \ldots, w_m) \in \R^m$ of point values associated with
the candidate positions in the preference orders.
Each voter gives each candidate the number of points associated with
this candidate's position, and the candidates with the highest total
score are the winners.  For example, the Plurality rule uses vectors
of the form $(1,0, \ldots, 0)$, the $k$-Approval rule uses vectors
with $k$ ones followed by $m-k$ zeros, and the Borda rule uses vectors
of the form $(m-1, \ldots, 2,1,0)$.

Given an election $E = (C,V,\{\succ_v\}_{v \in V})$, we sometimes
speak of a positional scoring rule $\cR_{(\bw^v)_{v \in V}}$, where
each voter has a separate scoring vector
$\bw^v = (w^v_1, \ldots, w^v_m)$.  This is particularly useful, for
example, to model weighted elections, where each voter $v$ has a
positive integer weight $\omega_v$ and is treated as $\omega_v$ copies
of a unit-weight voter; then, instead of using some rule $\cR_\bw$ and
incorporating weights directly into our algorithms, we can use rule
$\cR_{(\omega_v\cdot\bw)_{v \in V}}$. As an added benefit, our
algorithms become more general.  We will sometimes use
$\Delta w^v_\ell$ as a shorthand for $w^v_\ell - w^v_{\ell + 1}$.

The Copeland rule is based on the idea of pairwise elections among the
candidates. Let $E$ be an election and let $c, c'$ be two
candidates. By $N_E(c,c')$ we mean the number of voters who prefer $c$
over $c'$, i.e., $|V_{c \succ c'}|$.
We say that a candidate $c$ \emph{wins pairwise election}
against $c'$ if $N_E(c, c') > N_E(c', c)$. Similarly, we say that $c$
ties (resp. loses) pairwise election against $c'$ if
$N_E(c, c') = N_E(c', c)$ (resp. $N_E(c, c') < N_E(c', c)$).

For $\alpha \in [0, 1]$, the \emph{$\copeland$} rule assigns to each
candidate $c$ one point for each candidate with whom $c$ wins a
pairwise election, and $\alpha$ points for each candidate with whom
$c$ ties.  Formally, each candidate $c$ receives
$|\{c' \in C \setminus \{c\} : N_E(c, c') > N_E(c', c)\}| + \alpha
|\{c' \in C \setminus \{c\} : N_E(c, c') = N_E(c', c)\}|$
points. The winners are all the candidates with the maximum score.

For a voting rule $\cR$, we write $\scr_{E, \cR}(c)$ to denote the
score that candidate $c$ receives in election $E$.  We sometimes drop
the subscript $\cR$ when it is clear from the context.

\paragraph{Shift-Bribery.}
A \textsc{Shift\--Bribery} instance $I = (E, p,\psi)$
consists of an election $E = (C, V,$ ${\{\succ_v\}_{v \in V}})$, a
preferred candidate $p \in C$, and a collection $\psi = \{\psi_v\}_{v \in V}$
of the voters' price functions.
Each voter $v$ has the price function
$\psi_v \colon \{0\} \cup [\pi_v^{-1}(p) - 1] \to \R^+_0 \cup
\{\infty\}$ and $\psi_v(t)$
specifies the cost of shifting the preferred candidate forward by $t$
positions in $v$'s preference order.
We require that $\psi_v(0) = 0$ and that the function is
non-decreasing
($\psi_v(0) \leqs \psi_v(1) \leqs \cdots \leqs \psi_v(\pi^{-1}_v(p) -
1)$).
If $\psi_v(t) = \infty$ for some voter $v$ and value $t$, then it is
impossible to shift the preferred candidate by $t$ or more positions
in the preference order of $v$.
For an instance $I$, by $\psi^{\max}(I)$ we denote the highest
non-infinity price that occurs within $I$.
We write $\Delta\psi_v(\ell)$ and $|I|$ as
shorthands for $\psi_v(\ell) - \psi_v(\ell - 1)$ and $mn$,
respectively.

A \emph{shift action} $\bs = (s_v)_{v \in V}$ for an instance
$I = (E, p,\psi)$ of \textsc{Shift-Bribery} is a
vector of non-negative integers such that for each voter $v$ we have
$s_v < \pi_v^{-1}(p)$. Intuitively, this vector specifies for each
voter by how many positions we should shift the preferred candidate.
We say that shift action $\bs$ consists of $\sum_{v \in V} s_v$
\emph{unit shifts} and we define its cost to be
$\cost_I(\bs) = \sum_{v \in V} \psi_v(s_v)$.
We denote the election that results from applying $\bs$ to $E$ by
$\shf(E, \bs)$; formally, $\shf(E, \bs)$ has the same voters and
candidates as $E$ but for each voter $v$, if $v$'s preference order
induced function $\pi_v$ in $E$ then in $\shf(E, \bs)$ it induces
function $\pi'_v$ such that:
\begin{align*}
\pi'_v(j) =
\begin{cases}
\pi_v(j) &\text{ if } j < \pi_v^{-1}(p) - s_v,  \\
p & \text{ if } j = \pi_v^{-1}(p) - s_v, \\
\pi_v(j - 1) & \text{ if } \pi_v^{-1}(p) - s_v < j \leqs \pi_v^{-1}(p), \\
\pi_v(j) &\text{ if } j > \pi_v^{-1}(p). 
\end{cases}
\end{align*}

Let $\cR$ be a voting rule and let $I$ be a \textsc{Shift-Bribery}
instance with election $E$ and preferred candidate $p$.  
Shift action $\bs$ is \emph{successful} for $I$ under $\cR$ if $p$ is
an $\cR$-winner in $\shf(E, \bs)$, i.e., if
$p \in \cR(\shf(E, \bs))$.  \textsc{$\cR$-\textsc{Shift-Bribery}} is
an optimization problem where, given a \textsc{Shift-Bribery}
instance $I$, we ask for a successful shift action with the lowest
cost. We write $\opt(I)$ to denote this lowest cost.

\paragraph{Special Price Functions.}
There are two particularly interesting families of price functions.  A
\emph{unit price} function defines the cost of each unit shift to be
one, i.e., if $\psi_v$ is a unit price function then $\psi_v(\ell) =
\ell$ for each legal shift value $\ell$.  An \emph{all-or-nothing
  price} function is such that the cost of shifting the preferred
candidate is the same, irrespective by how many positions we shift him
or her.  Formally, if $\psi_v$ is an all-or-nothing price function
then there is a value $c_v$ such that $\psi_v(\ell) = c_v$ for all
positive integers $\ell$ that represent legal shifts (and, of course,
$\psi_v(0) = 0$).  An instance $I = (E, p, \psi)$ has \emph{$(1,
  \infty)$-all-or-nothing prices} if it has all-or-nothing price
functions and for each voter $v$ the value $c_v$ is in $\{1,
\infty\}$. Given such an instance $I$, we define its \emph{width} to
be the maximum of $\pi^{-1}_v(p) - 1$ over all $v \in V$ such that
$c_v = 1$. In other words, it is the maximum number of unit shifts
possible to perform within a single vote by paying a unit of
price. Another family of all-or-nothing prices that we will discuss is
the family of \emph{uniform-all-or-nothing prices}, for which $c_v =
1$ for all $v \in V$.

\paragraph{Linear Programming.}
In the \textsc{Linear Programming} problem we are given
an $m \times n$ matrix $A$, an $m$-dimensional
column vector $b$, an $n$-dimensional column vector $c$, and we ask
for an $n$ dimensional column vector $x$ that minimizes the value
$c^T x$ subject to the condition that $Ax \geqs b$. 
A \emph{basic solution}
to such a problem is a solution $x \in \R^n$ such that there are $n$
linearly independent rows $a_i$ of $A$ with $a_i x = b_i$. It is known
that when $\{x \in \R^n : Ax \geqs b\}$ is feasible and bounded, there always is
a basic solution that achieves the optimum, and it
can be
computed up to an arbitrary error in polynomial time
(see, e.g.,~\cite{LRM-book} for the use of basic solutions in approximation
algorithms).

\section{Borda Rule}\label{sec:borda}

We now move on to our results. We first show approximation schemes for
the case of the Borda rule, mostly focusing on the case of unit
prices.
We start with Borda because it is one of the simplest rules,
for which we can present our ideas most clearly, and because it is a very
practical rule (in particular, relevant to various competitions).

\subsection{Initial Observations}
We first define two values that will guide our algorithms, and we explain
their usefulness.

\begin{definition} For an instance $I = (E, p, \psi)$  of
  Borda-\textsc{Shift-Bribery} and a non-negative integer  $k$, we define:
  \begin{align*}
    \diffmax(I) &= \max_{c \in C}( \scr_E(c) - \scr_E(p)), \text{ and} \\
    \scrdiff(I,k) &= \sum_{c \in C} \max\{0, \scr_E(c) - \scr_E(p) -  k\}.
  \end{align*}
\end{definition}

The former value gives the score difference between the preferred
candidate and his or her strongest opponent, whereas the latter
measures the total number of points that the non-preferred candidates
need to lose, provided that the preferred one gains $k$ points.

Elkind et al.~\cite{EFS09} note that given an instance $I$ of
Borda-\textsc{Shift-Bribery}, if $K$ is the smallest number of unit
shifts in an optimal solution, then
$\diffmax(I)/2 \leqs K \leqs \diffmax(I)$. Indeed, if the preferred
candidate gains $\diffmax(I)$ points then he or she certainly matches
his or her strongest opponent.  On the other hand, the preferred
candidate needs at least $\diffmax(I)/2$ unit shifts because each of
them decreases the score difference between him or her and the
strongest opponent by at most two.  However, it turns out that
$\scrdiff(I,k)$ provides an even more useful bound.

\begin{lemma} Let $I = (E, p, \psi)$ be an instance of
  Borda-\textsc{Shift-Bribery}, let $\bs$ be a successful shift action
  for $I$, and let $k_{\bs}$ be the number of unit shifts within
  $\bs$.  Then, $\scrdiff(I,k_{\bs}) \leqs k_{\bs}$.
\end{lemma}
\begin{proof}
  After applying $\bs$, the score of $p$ is exactly
  $\scr_E(p) + k_{\bs}$. Thus each candidate $c \in C$ must have lost
  at least $\max\{0, \scr_E(c) - \scr_E(p) - k_{\bs}\}$ points. This
  indeed means that
  $k_{\bs} \geqs \sum_{c \in C} \max\{0, \scr_E(c) - \scr_E(p) -
  k_{\bs}\} = \scrdiff(I,k_{\bs})$, as desired.
\end{proof}

We will also make use of the following subroutine,
which is based on a simple dynamic program.
(Note that it is not restricted to unit prices.)

\begin{lemma}\label{lem:dp}
  There exists an algorithm that given an instance $I = (E, p, \psi)$
  of Borda-\textsc{Shift-Bribery}, a subset
  $C' = \{c_1, \ldots, c_t\}$ of $t$ candidates from $E$, and a vector
  $(s_1, \ldots, s_t)$ of non-negative integers, computes a minimum cost shift
  action that ensures that each candidate $c_i$ loses at least $s_i$
  points. The algorithm runs in polynomial time with respect to
  $|I|+\prod_{i=1}^t(s_i+1)$.
\end{lemma}

\begin{proof}
  Let the notation be as in the statement of the lemma, let $m$
  be the number of candidates in $E$, and let
  $V = (v_1, \ldots, v_n)$ be the collection of voters in $E$.

  We use dynamic programming with table $\mathit{DP}$ of dimension
  $(n + 1) \times (nm+1) \times (s_1+1) \times \cdots \times (s_t+1)$.  The
  $(i, j, r_1, \dots, r_t)$-entry of the table stores the minimum cost
  of $j$ unit shifts on voters $v_1, \dots, v_i$ that jointly
  decrease the score of each $c_i$ by at least $r_i$ points. In the
  beginning, all the entries are set to $\infty$, except the
  $(0, 0, \dots, 0)$-entry, which is set to zero.

  Then we go through $i = 1, \dots, n$ in the increasing order and for
  every entry of the form $(i, j, r_1, \ldots, r_t)$, we compute its
  value using the following formula:
  \begin{align*}
    &\mathit{DP}[i][j][r_1] \cdots [r_t] =\\
    &\hspace{80pt} \min_{\ell \leqs \min\{j, \pi^{-1}_{v_i}(p)-1\}} \Big(\psi_{v_i}(\ell) + \mathit{DP}[i-1][j - \ell][\max\{0, r_1 - b_1\}] \cdots [\max\{0, r_t - b_t\}]\Big),
  \end{align*} 
  where each
  $b_z := \ind[0 < \pi^{-1}_{v_i}(p) - \pi^{-1}_{v_i}(c_z) \leqs \ell]$
  is the indicator variable specifying whether shifting the preferred
  candidate $p$ by $\ell$ positions within vote $v_i$ causes him or
  her to pass $c_z$.

  Finally, $\mathit{DP}[n][nm][s_1] \cdots [s_t]$ gives the minimum cost of
  the desired shift action.
\end{proof}

\subsection{A Combinatorial PTAS for Unit Prices}
We now  give a simple combinatorial PTAS for Borda-\textsc{Shift-Bribery}
with unit prices.
The main idea of our algorithm is as follows.  If the optimal number
of unit shifts needed is $\opt$, then, in total, the scores of other
candidates decrease by at most $\opt$. This means that, once we guess
$\opt$ correctly, for each $\varepsilon > 0$ there can be at most
$1/\varepsilon$ ``bad'' candidates, whose scores exceed that of the
preferred candidate by more than $(1 + \varepsilon)\opt$. Since there
are only $1/\varepsilon$ such candidates, we can use the algorithm
from Lemma~\ref{lem:dp} to compute the cheapest set of (at most)
$\opt$ unit shifts that ensure that the preferred candidate defeats
these candidates.  Then, we shift the preferred candidate up further
$\varepsilon \opt$ times, which ensures that $p$ also defeats all the
other candidates. In total, we use only $(1 + \varepsilon)\opt$ shifts
and hence we arrive at our PTAS for the unit prices case. This idea is
formalized below.

\begin{theorem} \label{thm:simple-ptas-borda} For each
  $\varepsilon > 0$, there exists an algorithm that given an instance
  $I$ of Borda-\textsc{shift-bribery} with unit prices runs in time
  $\opt(I)^{O(1/\varepsilon)} \poly(|I|)$ and outputs a successful
  shift action of cost at most $(1 + \varepsilon)\opt(I)$.
\end{theorem}

\begin{proof}
  Let $I = (E, p, \psi)$ be an instance of
  Borda-\textsc{Shift-Bribery} and let $\varepsilon > 0$ be the
  desired approximation ratio. 
  For every $k$  between $\diffmax(I)/2$ and $\diffmax(I)$, such that
  $\scrdiff(I,k) \leqs k$, we
  execute the following steps:
  \begin{enumerate}
  \item Let $C_{\bad}^k = \{c_1, \ldots, c_{t(k)}\}$ be the set
    of candidates whose scores are greater than
    $\scr_E(p) + (1 + \varepsilon)k$. We use the algorithm from 
    Lemma~\ref{lem:dp} to find the least-cost shift action that
    decreases the score of each $c \in C_{\bad}^k$ by at least
    $\scr_E(c) - \scr_E(p) - k$ points. \label{step:dp2}

  \item If the cost of this shift action is at most $k$, then we
    perform additional arbitrary unit shifts so that the total number
    of unit shifts is $\lfloor (1 + \varepsilon)k \rfloor$ or, if not
    enough unit shifts are possible, we shift $p$
    to the top of every vote.
    We output all the performed unit shifts and
    terminate.\label{step:output}
  \end{enumerate}

  We first note that the algorithm indeed outputs a successful shift
  action.  If $p$ ends up being on top of all the votes then he or she
  clearly wins. On the other hand, if the total number of unit shifts
  performed is $\lfloor (1 + \varepsilon)k \rfloor$, then the score of
  $p$ is at least $\scr_E(p) + \lfloor (1 + \varepsilon)k \rfloor$;
  this means that, for all the candidates $c \notin C_{\bad}^k$, the new
  score of $p$ is at least $\scr_E(c)$, which is at least as large as
  the score of $c$ after the shifts. Moreover, the algorithm from
  Lemma~\ref{lem:dp} ensures that after the shifts the score of $p$ is
  at least as high as the scores of all the candidates from
  $C_{\bad}^k$. Thus, $p$ is a winner.

  Next, let us argue that the algorithm computes a
  \mbox{$(1+\varepsilon)$}-approximate solution. 
  Recall that due to the results of Elkind et al.~\cite{EFS09}, 
  the number of unit shifts in the optimal
  solution is between $\diffmax(I)/2$ and $\diffmax(I)$.
  Therefore the algorithm must
  terminate at latest when considering $k = \opt(I)$.
  Given this many
  shifts, it is---by definition---possible for $p$ to obtain score higher
  than all the candidates from $C_{\bad}^k$ and, so, the algorithm from
  Lemma~\ref{lem:dp} returns a shift action with at most $\opt(I)$ unit
  shifts. Thus the algorithm terminates with at most
  $\lfloor (1 + \varepsilon)\opt(I) \rfloor$ unit shifts.

  The running time of the algorithm follows from Lemma~\ref{lem:dp}
  and is bounded by a polynomial in $|I|$ and 
  $\Pi_{c \in C_\bad^k}(\scr_E(c)-\scr_E(p)-k + 1) \leqs
  (\diffmax(I)+1)^{t(k)} = \opt(I)^{O(t(k))}$.
  However, we only invoke
  Lemma~\ref{lem:dp} when $k \geqs \scrdiff(I,k)$. This means that:
  \begin{align*}
    k \geqs \scrdiff(I,k) = \sum_{c \in C} \max\{0, \scr_E(c) - \scr_E(p) - k\} \geqs \sum_{c \in C_\bad^k} (\scr_E(c) - \scr_E(p) - k) > t(k) \cdot \varepsilon k,
  \end{align*}
  and we conclude that $t(k) < 1/\varepsilon$.
  Thus the running time is
  polynomial with respect to $|I| + \opt(I)^{O(1/\varepsilon)}$.
\end{proof}

\subsection{A Faster FPT Algorithm}
Using a very similar reasoning as in Theorem~\ref{thm:simple-ptas-borda},
we obtain an $\FPT$ algorithm for Borda-\textsc{Shift-Bribery} (with
arbitrary price functions) parameterized by the smallest number $K$ of
unit shifts in an optimal solution. 
While this case was already known
to be in $\FPT$, our algorithm runs in $2^{O(K)}\poly(|I|)$ time, which
is faster than the known $2^{O(K^2)}\poly(|I|)$-time algorithm of
Bredereck et al.~\cite{BCFNN16}.

\begin{theorem} \label{thm:fpt-shift} There exists a
  $2^{O(K)} \poly(|I|)$-time algorithm that can find an optimal
  solution to a given instance $I = (E, p, \psi)$ of
  Borda-\textsc{Shift-Bribery} where $K$ is 
  the smallest number of unit shifts in an optimal solution.
\end{theorem}

\begin{proof}
  As in the proof of Theorem~\ref{thm:simple-ptas-borda}, our algorithm
  tries all the values of $k$ between $\diffmax(I)/2$ and
  $\diffmax(I)$ such that $\scrdiff_E(I,k) \leqs k$. For such a $k$ it
  proceeds as follows: (a) It forms the set
  $C^k_\bad = \{c_1, \ldots, c_{t(k)}\}$ of all the candidates whose scores
  exceed $\scr_E(p)+k$ and (b) using the algorithm from
  Lemma~\ref{lem:dp}, it computes the cheapest shift bribery that
  consists of $k$ unit shifts and ensures that each candidate $c$ in
  $C^k_\bad$ loses at least $\scr_E(c)-\scr_E(p)-k$ points. If such a
  shift action exists then it stores it (note that all stored shift
  actions are successful). The algorithm outputs the cheapest stored
  shift action.

  The correctness of the algorithm follows from the fact that when it
  considers the value of $k = K$, then the algorithm from
  Lemma~\ref{lem:dp} indeed finds an optimal shift action.

  The running time of the algorithm is polynomial in $|I|$ and the
  values $\prod_{i=1}^{t(k)} (\scr_E(c_i) - \scr_E(p) - k + 1)$ for the
  values of $k$ for which Lemma~\ref{lem:dp} is invoked.
  When
  $\scrdiff_E(k) \leqs k$ then the following holds (the first
  inequality is due to the AM-GM inequality\footnote{The AM-GM
    inequality says that for non-negative numbers $x_1, \ldots, x_n$
    we have
    $\frac{1}{n}(x_1+ \cdots + x_n) \geqs \sqrt[n]{x_1x_2 \cdots
      x_n}$.}
  and the last one follows from the Bernoulli's
  inequality\footnote{For non-negative $x$ and $r$ between $0$ and $1$
    we have $(1+x)^r \leqs 1 + rx$. In our case $x = \nicefrac{k}{t(k)}$
    and $r = \nicefrac{t(k)}{k}$; $k \geqs t(k)$ because we need at least
    one unit shift
    for each candidate in $C_\bad$.}):
  \begin{align*}
    &\prod_{i=1}^{t(k)} (\scr_E(c_i) - \scr_E(p) - k + 1) \leqs \left(\frac{1}{t(k)} \cdot{\sum_{i=1}^{t(k)} (\scr_E(c_i) - \scr_E(p) -k + 1)}\right)^{t(k)} \\
                  &= \left(\frac{1}{t(k)} \cdot{\scrdiff_E(k) + 1}\right)^{t(k)}
                  \leqs \left(1 + \frac{k}{t(k)}\right)^{t(k)}
                  = \left(\left(1 + \frac{k}{t(k)}\right)^{\frac{t(k)}{k}}\right)^k
                  \leqs 2^k 
                  \leqs 2^{\diffmax(I)}
                  \leqs 2^{2K}.
  \end{align*}
  Thus the running time of our algorithm is polynomial in $|I|$ and $2^{O(K)}$,
  which suffices to complete the proof.
\end{proof}

\subsection{EPTAS for Unit Prices}
The main result of this subsection is an EPTAS (efficient
polynomial-time approximation scheme) for Borda-\textsc{Shift-Bribery}
with unit prices, that is, a PTAS for which the non-polynomial factors
in the running time depend only on the required approximation
ratio. Note that in the algorithm from
Theorem~\ref{thm:simple-ptas-borda} this factor was
$\opt(I)^{O(1/\varepsilon)}$ and, thus, did not depend on
$\varepsilon$ alone.

On the technical level, we first develop an algorithm that for an
instance $I$ of Borda-\textsc{Shift-Bribery} with unit prices outputs
a solution with cost at most $\opt(I) + \sqrt{\opt(I)}$.
The general structure of our algorithm is similar to that of the
algorithm from Theorem~\ref{thm:simple-ptas-borda}, but instead of invoking
Lemma~\ref{lem:dp}, we solve a linear program.  We form this program
in such a way that its basic solution has to consist almost entirely
of integral values. Then, 
rounding and complementing the
obtained shift action with arbitrary unit shifts gives the desired
solution.

To state our linear program, we will model shift actions as boolean
matrices $(x_{(v, j)})_{v \in V, j \in [m]}$ such that $x_{(v,j)}$ is
$1$ if after applying the shift action the preferred candidate is
ranked on position $j$ or better in vote $v$, and it is $0$ otherwise
(so we will always have
$0 \leqs x_{(v,1)} \leqs x_{(v,2)} \leqs \cdots \leqs x_{(v,m)} \leqs~1$).
Formally, a shift action $\bs = (s_v)_{v \in V}$ corresponds to the
boolean matrix $x_{(v, j)} = \ind[\pi_v^{-1}(p) - s_v \leqs j]$.

\begin{lemma} \label{lem:aptas-borda-unit} There is a polynomial-time
  algorithm that, given an instance $I$ of
  Borda-\textsc{Shift-Bribery} with unit prices, outputs a successful
  shift action with cost at most $\opt(I) + \sqrt{\opt(I)}$.
\end{lemma}

\newcommand{\LPU}{\mathrm{LP\hbox{-}U}}

\begin{figure}
  \small
  \FrameSep=3pt
  \begin{framed}
  \begin{align}
    &\text{minimize }   \sum_{v \in V} \sum_{j \in [\pi_v^{-1}(p) - 1]} x_{(v, j)}  \text{\quad s.t.: }  \nonumber \\
    & 0 \leqs x_{(v, 1)} \leqs \cdots \leqs x_{(v, \pi_v^{-1}(p) - 1)} \leqs 1 \text{\quad\quad ,} \forall v \in V \label{eq:shift-order1} \\
    & \!\!\!\!\!\!\!\!\! \sum_{v \in V_{c \succ p}} x_{(v, \pi_v^{-1}(c))}  \geqs \scr_E(c) - \scr_E(p) - k  \text{\quad\quad ,}\forall c \in C_{\bad}^k \label{eq:score-dec1}
  \end{align}
  \end{framed} 
  \caption{\label{fig:lpu}Program $\LPU(I,k)$ from the proof of
    Lemma~\ref{lem:aptas-borda-unit}. For each voter $v$, we have
    variables $x_{(v,1)},$ $\ldots, x_{(v,\pi_v^{-1}(p)-1)}$.
    Constraints~\eqref{eq:shift-order1} ensure that an integral
    solution describes a valid shift action
    and Constraints~\eqref{eq:score-dec1} ensure that after applying such an
    action, each candidate in $C_\bad^k$ has score no higher than $p$;
    recall that $V_{c \succ p}$
    means the set of voters that prefer $c$ to $p$. The optimization
    goal is to minimize the number of unit shifts in the shift
    action.}
\end{figure}

\begin{proof}
  Let $I = (E, p, \psi)$ be an instance of
  Borda-\textsc{Shift-Bribery} with unit prices, where
  $E = (C, V, \{\succ_v\}_{v \in V})$.  We try all integers $k$
  between $\diffmax(I)/2$ and $\diffmax(I)$, such that
  $\scrdiff(I,k) \leqs k$, and for each of them we perform the following
  steps:
  \begin{enumerate}
  \item We form the set $C_\bad^k \subseteq C$ of those candidates $c$
    whose scores exceed value $\scr_E(p) + k + \sqrt{k}$.
  \item We form the linear program $\LPU(I,k)$ from
    Figure~\ref{fig:lpu} and solve it for a basic solution (note that
    the objective function gives the number of unit shifts used). If
    $\LPU(I,k)$ is infeasible or the cost of the solution exceeds $k$,
    then we skip this value of $k$.  Let
    $(x^{\opt}_{(v, j)})_{v \in V, j \in [\pi^{-1}_v(p) - 1]}$ be the
    basic optimal solution found.
  \item Let $\bs^{\mathit{LP}}$ be the shift action corresponding to
    $(\lfloor x^{\opt}_{(v, j)} \rfloor)_{v \in V, j \in
      [\pi^{-1}_v(p) - 1]}$.
    (Note that the rounded solution indeed correctly describes a shift
    action and that its cost, i.e., the number of unit shifts it
    contains, is at most $k$.)  Form a shift action $\bs$ by extending
    $\bs^{\mathit{LP}}$ so that it contains $k + \lfloor \sqrt{k} \rfloor$ unit
    shifts or, if this is impossible, then so that $p$ is on the top
    of every preference order. Output $\bs$ and
    terminate. \label{step2:add-and-output}
  \end{enumerate}

  Since finding basic solutions for linear programs can be done in
  polynomial time, the whole algorithm runs in polynomial time.
  Further, the cost of the computed shift action $\bs$ is at most
  $\opt(I) + \sqrt{\opt(I)}$ unit shifts.  To see this, consider the
  step when the algorithm tries $k = \opt(I)$; if it terminates
  earlier then our claim certainly is satisfied.  For this value of
  $k$, $\LPU(I,k)$ certainly has a solution of cost at most $k$
  because an optimal successful shift action for $I$ is one of its
  feasible solutions. Thus the algorithm terminates for this value of
  $k$ and (in Step~\ref{step2:add-and-output}) outputs a shift action
  with at most $\opt(I) + \sqrt{\opt(I)}$ unit shifts.

  It remains to show that the computed shift action $\bs$ is
  successful. If $p$ ends up at the top of every preference order,
  then surely $\bs$ is a successful shift action. Let us consider the
  case where $\bs$ consists of exactly $k + \lfloor \sqrt{k} \rfloor$
  unit shifts, i.e., where after applying $\bs$, $p$ has score
  $\scr_E(p) + k + \lfloor \sqrt{k} \rfloor$, for the value of $k$ for
  which the algorithm terminates. Since prior to applying $\bs$ each
  candidate $c \notin C_{\bad}^k$ had score at most
  $\scr_E(p) + k + \lfloor \sqrt{k} \rfloor$, after applying $\bs$
  candidate $p$ certainly has score at least as high as theirs.  Thus we only
  need to show that each candidate in $C_\bad^k$ also ends up with score
  at most $\scr_E(p) + k + \lfloor \sqrt{k} \rfloor$.

  Let us say that a voter $v \in V$ is \emph{integral} if
  $x_{(v, j)}^{\opt} \in \{0, 1\}$ for all $j \in [\pi^{-1}_v(p) - 1]$
  and let $\Vi$ be the set of all integral voters.  Since
  $(x^{\opt}_{(v, j)})_{v \in V, j \in [\pi_v^{-1}(p) - 1]}$ is a
  basic solution of $\LPU(I,k)$, at least
  $\sum_{v \in V} (\pi_v^{-1}(p) - 1)$ inequalities must be
  tight. For each integral voter
  $v$, exactly $\pi_v^{-1}(p) - 1$ inequalities of the
  form~\eqref{eq:shift-order1} are tight. On the other hand, for each
  non-integral voter $v$, at most $\pi_v^{-1}(p) - 2$ inequalities
  in~\eqref{eq:shift-order1} are tight.  Further, there are $|C_\bad^k|$
  inequalities of the form~\eqref{eq:score-dec1} and
  $|C_\bad^k| \leqs \lfloor\sqrt{k} \rfloor$, because each candidate
  $c \in C_\bad^k$ contributes at least $\sqrt{k}$ points to
  $\scrdiff(I,k)$ and $\scrdiff(I,k) \leqs k$.  Altogether, this means
  that there are at most $\lfloor \sqrt{k} \rfloor$ non-integral
  voters, i.e., $|V \setminus \Vi| \leqs \lfloor \sqrt{k}
  \rfloor$.
  Intuitively, each tight inequality of the form~\eqref{eq:score-dec1}
  can lead to at most a single non-integral voter.

  Rounding the basic solution into $\bs^{\mathit{LP}}$
  can increase score difference
  between $c$ and $p$ by at most 1 for each non-integral voter (score of $c$ is
  increasing and score of $p$ is decreasing).
  Hence, in total, after rounding the difference can be at most
  $|V \setminus \Vi| \leqs \lfloor \sqrt{k}\rfloor$.
  More formally, for each $c \in C_{\bad}^k$, the score of $c$ after applying
  $\bs^{\mathit{LP}}$ is as follows
  (for a number $x$, $0 \leqs x \leqs 1$, by $\fr(x)$ we
  mean its fractional part,
  so that if $0 \leqs x < 1$ then $\fr(x) = x$ and if $x=1$ then 
  $\fr(x) = 0$):
  \begin{align*}
    &\scr_E(c) - \sum_{v \in V_{c \succ p}} \Big\lfloor x^{\opt}_{(v, \pi_v^{-1}(c))} \Big\rfloor
    = \scr_E(c) - \sum_{v \in V_{c \succ p}} x^{\opt}_{(v, \pi_v^{-1}(c))} + \sum_{v \in V_{c \succ p}} \fr\Big(x^{\opt}_{(v, \pi_v^{-1}(c))}\Big) \\
    &\overset{\eqref{eq:score-dec1}}{\leqs} \scr_E(p) + k + \sum_{v \in V_{c \succ p}} \fr\Big(x^{\opt}_{(v, \pi_v^{-1}(c))}\Big)
    = \scr_E(p) + k + \sum_{v \in V_{c \succ p} \setminus \Vi} \fr\Big(x^{\opt}_{(v, \pi_v^{-1}(c))}\Big) \\
    &< \scr_E(p) + k + |V \setminus \Vi| 
    \leqs \scr_E(p) + k + \lfloor \sqrt{k} \rfloor. 
  \end{align*}
  This means that $p$ is indeed a winner of the election.
\end{proof}

Using this algorithm and the combinatorial PTAS from
Theorem~\ref{thm:simple-ptas-borda}, we obtain our EPTAS (in short, if
$\diffmax(I) < 2/\varepsilon^2$ then we run the algorithm from
Theorem~\ref{thm:simple-ptas-borda} and we run the
algorithm from Lemma~\ref{lem:aptas-borda-unit} otherwise).

\begin{minipage}{\textwidth}
\begin{theorem} \label{thm:eptas-borda-unit} There is an algorithm
  that, given an instance $I$ of Borda-\textsc{Shift-Bribery} with
  unit prices and a positive 
  number $\varepsilon > 0$, runs in
  time $2^{O(\log(1/\varepsilon)/\varepsilon)} \poly(|I|)$ and outputs
  a successful shift action of cost at most $(1 + \varepsilon)\opt(I)$.
\end{theorem}
\end{minipage}

\begin{proof}
  Given an instance $I$ of Borda-\textsc{Shift-Bribery} with unit
  prices, our algorithm proceeds as follows: If
  $\diffmax(I) < 2/\varepsilon^2$ then it runs the algorithm from
  Theorem~\ref{thm:simple-ptas-borda} and, otherwise, if
  $\diffmax(I) \geqs 2/\varepsilon^2$, it runs the algorithm from
  Lemma~\ref{lem:aptas-borda-unit}. Then it returns the output from
  the respective algorithm.

  If $\diffmax(I) < 2/\varepsilon^2$, then
  $\opt(I) < 2/\varepsilon^2$, which means that the algorithm from
  Theorem~\ref{thm:simple-ptas-borda} runs in time
  $\opt(I)^{O(1/\varepsilon)}\poly(|I|) =
  2^{O(\log(1/\varepsilon)/\varepsilon)}\poly(|I|)$.
  Since the algorithm from Lemma~\ref{lem:aptas-borda-unit} always
  runs in time $\poly(|I|)$, we conclude that the above algorithm runs
  in time $2^{O(\log(1/\varepsilon)/\varepsilon)}\poly(|I|)$.

  When we run the algorithm from Theorem~\ref{thm:simple-ptas-borda},
  the solution is always a $(1 + \varepsilon)$-approximation. On the
  other hand, when we invoke Lemma~\ref{lem:aptas-borda-unit}, we have
  $\opt(I) \geqs \diffmax(I)/2 \geqs 1/\varepsilon^2$; hence, the
  output is of cost at most
  $\opt(I) + \sqrt{\opt(I)} = (1+\opt(I)^{-0.5})\opt(I)  \leqs (1 + \varepsilon) \opt(I)$.
  This concludes the proof.
\end{proof}

Theorem~\ref{thm:eptas-borda-unit} gives formal evidence that
approximating Borda-\textsc{Shift-Bribery} is computationally easier
for the case of unit prices than for the general case or, even, for
the all-or-nothing prices case. The latter cases are 
$\W[1]$-hard when parameterized by the budget~\cite{BCFNN16} and this
means that no EPTAS exists for them unless $\W[1] = \FPT$. If there
were an EPTAS for Borda-\textsc{Shift-Bribery} for the case of general
prices, or for the all-or-nothing prices, which ran in time
$f(\varepsilon) n^{O(1)}$, then we could plug in
$\varepsilon = 1/(2B)$, where $B$ would be the budget limit, and solve
the problem exactly in time $f(\frac{1}{2B}) n^{O(1)}$, implying that
$\W[1] = \FPT$. Such connections between 
$\FPT$ algorithms and EPTASes are well-known in theoretical computer
science~\cite{CT97}, but, so far, have not found many applications in
computational social choice.

\subsection{Uniform-All-or-Nothing Prices}

In addition to the above PTASes, we devise a simple greedy algorithm
that yields an asymptotic 1.5-approximation ratio for the special case
of uniform-all-or-nothing prices. While this will be subsumed by our
PTAS in the next section, the simplicity of the algorithm may make it
more practical than the LP-based PTAS in
Theorem~\ref{thm:ptas-borda-main}.  The main idea of our greedy
algorithm is to simply shift the preferred candidate $p$ to the top in
the votes where $p$ is ranked lowest. The idea is formalized below.

\begin{theorem} \label{thm:apx-borda-uniform-aon} There is a greedy
  algorithm that, given a Borda-\textsc{Shift-Bribery} instance $I$
  with uniform-all-or-nothing prices, outputs a successful shift
  action with cost at most $(1.5 \cdot \opt(I) + 1)$ in polynomial
  time.
\end{theorem}

\begin{proof}
  Let $I = (E, p, \psi)$ be an instance of
  Borda-\textsc{Shift-Bribery} with uniform-all-or-nothing prices,
  where $E = (C, V, \{\succ_v\}_{v \in V})$.  Let $m$ be the number of
  candidates. Recall that for all-or-nothing prices a solution is
  simply a set of bribed voters because we can move $p$ to the tops of
  their preference orders without additional cost.

  Our algorithm proceed as follows: As long as $p$ is still not a
  winner, the algorithm bribes a voter that ranks $p$ on the lowest
  position. When the lowest position is at least three, ties are
  broken arbitrarily. However, if the lowest position is two, then the
  algorithm breaks the ties by choosing the voter whose top-ranked
  candidate has the highest score.

  For each $j$, let $V_j \subseteq V$ denote the set of the voters
  that rank $p$ on position $j$.  The key observation for the analysis
  of our algorithm is that to obtain an optimal solution, it suffices
  to bribe all the voters from $V_m$, then all the voters from
  $V_{m-1}$, and so on, until some value $i$, such that we only bribe
  a subset of the voters from $V_i$. We formalize this observation in
  the following lemma.

  \begin{lemma}\label{lem:opt-w-exists}
    There exists an integer $i \geqs 2$, an optimal solution $W$ for
    $I$, and a set of voters $W_i \subseteq V_i$ such that $W =
    (\bigcup_{j = i+1}^m V_j) \cup W_i$.
  \end{lemma}
  \begin{subproof}
    Let $W'$ be some optimal solution for $I$. As long as $W'$ is
    still not in the required form, we modify $W'$ as follows:
    \begin{enumerate}
      \item remove a voter in $W'$ that ranks $p$ on the highest position,
      \item add a voter from $V \setminus W'$ that ranks $p$
        on the lowest position.
    \end{enumerate}
    Note that in each step of this process the removed voter and the
    added voter are from different $V_i$'s.  Thus in each step we
    increase the final scores of some candidates by at most $1$ (by
    canceling a shift which decreased their scores) but we also
    increase the final score of $p$ by at least $1$.  Hence $W'$
    remains an optimal solution after each step.  It is also easy to
    see that the process terminates when a solution is in the required
    form.
  \end{subproof}

  Clearly, the greedy algorithm also first bribes all the voters in
  $(\bigcup_{j = i+1}^m V_j)$ and it never bribes any voter in $V_{i -
    1}$. The only difference between the greedy algorithm and $W$ is
  in how they bribe voters in $V_i$. To analyze this difference, let
  $E'$ denote the election after all voters in $(\bigcup_{j = i+1}^m
  V_j)$ are bribed. Note that in $E'$ each voter ranks $p$ among his
  or her top $i$ candidates.

  We consider two cases based on the value of $i$:

  \begin{description}
  \item[Case $\boldsymbol{i = 2}$.] In this case the greedy algorithm
    outputs an optimal solution because it bribes the voters in $V_2$
    in an optimal order (the score of the candidate with the highest
    score is decreased first, and so on).
  \item[Case $\boldsymbol{i \geqs 3}$.] We claim that in this case the
    algorithm outputs a solution of cost at most $\frac{i}{i - 1}
    \opt(I) + 1$; note that this immediately implies the approximation
    guarantee in Theorem~\ref{thm:apx-borda-uniform-aon}.
    To see this, observe that since $W$ only bribes $|W_i|$ voters in
    $E'$, every candidate $c \ne p$ has his score decreased by at most
    $|W_i|$ points, whereas the score of our preferred candidate $p$
    increases by exactly $(i - 1) \cdot |W_i|$ points. Since $p$ wins
    the election after the voters in $W_i$ are bribed, we have
    $\scr_{E'}(c) - \scr_{E'}(p) \leq i \cdot |W_i|$.
    Since every voter bribed in $E'$ by the greedy algorithm increases
    the score of $p$ by $i - 1$, $p$ becomes a winner after at most
    $\lceil \frac{i}{i - 1}\cdot|W_i| \rceil$ voters in $E'$ are
    bribed. In total, the greedy algorithm bribes at most
    $\sum_{j=i+1}^m |V_j| + \lceil \frac{i}{i - 1}\cdot|W_i| \rceil
    \leqs \frac{i}{i - 1}|W| + 1 = \frac{i}{i - 1}\opt(I) + 1$ voters
    in $E$, as claimed.
  \end{description}

  Note that the greedy algorithm gives a solution of cost at most
  $\frac{i}{i - 1}\opt(I) + 1$ for $i$ defined in
  Lemma~\ref{lem:opt-w-exists}.  Hence, the worst guarantee, $1.5 \cdot
  \opt(I) + 1$, is achieved for $i=3$.
\end{proof}

We conclude this section by remarking that the
Borda-\textsc{Shift}-\textsc{Bribery} problem with
uniform-all-or-nothing prices was not known to be NP-hard before;
specifically, Elkind et al.~\cite{EFS09} only showed NP-hardness for
$(1, \infty)$-all-or-nothing prices. Nevertheless, it is possible to
(carefully) modify the proof to yield NP-hardness for
uniform-all-or-nothing prices as well. We provide the full proof in
Appendix~\ref{app:np-hardness-uni}.

\section{Positional Scoring Rules}\label{sec:positional}

In this section we give our main result: A PTAS for the case of
\textsc{Shift-Bribery} with an arbitrary positional scoring rule,
whose scoring vectors are, possibly, different for different voters,
and for arbitrary prices. The algorithm and the proof is somewhat
involved and we split its description into two parts, by first
deriving an algorithm with an additive error and then using it to form
the desired PTAS.

\subsection{An Algorithm with Additive Error}
The crucial part of our algorithm
is an approximation algorithm
that yields a good solution in the case where
$\psi^{\max}(I)$, the highest non-infinite price in the instance, is small.

The main complication in the general prices case, as opposed to the
unit prices case, is that the cost of obtaining some $k+1$ points for
the preferred candidate can be far larger than the cost of obtaining
$k$ points. Thus the main trick used in the proofs from the previous
section---deciding up front how many more points than in an optimal
solution the preferred candidate would get---cannot be directly
applied.  We work around this problem by first solving a linear
program which, roughly speaking, for a given value $\varepsilon > 0$
tells us how many extra points the preferred candidate needs to obtain
so that he or she has score higher than all but at most
$\nicefrac{1}{\varepsilon}$ candidates. Then, using a technique
similar to the one we used for Lemma~\ref{lem:aptas-borda-unit}---in
particular, solving a second linear program, for which a basic
solution contains a large number of integral variables---we find our
approximate solution.

\begin{lemma} \label{lem:ptas-borda-main} There is an algorithm that
  given $\varepsilon > 0$ and an instance $I$ of
  $\cR$-\textsc{Shift-Bribery}, where $\cR$ is a given positional
  scoring rule with a possibly different scoring vector for each
  voter, outputs a successful shift action for $I$ of cost at most
  $(1 + \varepsilon)\opt(I) + (1 + 1/\varepsilon) \psi^{\max}(I)$ and
  runs in time $|I|^{O(1)}$.
\end{lemma}

\begin{figure}
  \small
  \FrameSep=3pt
  \begin{framed}
  \begin{align}
    &\text{minimize } \sum_{v \in V} \sum_{j \in [\pi_v^{-1}(p) - 1]} \Delta\psi_{v}(\pi_v^{-1}(p) - j) \cdot x_{(v, j)} \text{\quad s.t.: } \nonumber \\
    & 0 \leqs x_{(v, 1)} \leqs \cdots \leqs x_{(v, \pi_v^{-1}(p) - 1)} = x_{(v, \pi_v^{-1}(p))} = \cdots = x_{(v, m)} = 1  \text{\hspace{47pt} ,} \forall v \in V \label{eq:sstar} \\
         &\scr_E(c) - \sum_{v \in V_{c \succ p}} \Delta w_{\pi_v^{-1}(c)}^v \cdot x_{(v, \pi_v^{-1}(c))} \leqs  \scr_E(p) + \sum_{v \in V} \sum_{j \in [\pi_v^{-1}(p) - 1]} \Delta w^v_j \cdot x_{(v, j)}  \text{ ,} \forall c \in C \label{eq:more-score-1} \!\!\!
  \end{align}
  \end{framed}
  \caption{\label{fig:lp1}LP1 for the proof of
  Lemma~\ref{lem:ptas-borda-main}.  For each voter $v$, we have
  variables $x_{(v,1)}, \ldots, x_{(v,m)}$.  For an integral solution,
  Constraints~\eqref{eq:sstar} ensure that the variables specify a
  shift action, Constraints~\eqref{eq:more-score-1} ensure that 
  this shift action is
  successful, and the optimization goal specifies its cost.}
\end{figure}

\begin{figure}[t!]
  \small
  \FrameSep=3pt
  \begin{framed}
    \begin{align}
      &\text{minimize } \sum_{v \in V} \sum_{j \in [\pi_v^{-1}(p) - 1]} \Delta\psi_{v}(\pi_v^{-1}(p) - j) \cdot y_{(v, j)} \text{\quad s.t.:} \nonumber \\
       & 0 \leqs y_{(v, 1)} \leqs \cdots \leqs y_{(v, j_v - 1)} \leqs y_{(v, j_v)} = \cdots = y_{(v, m)} = 1 \text{\hspace{109pt} ,} \forall v \in V \label{eq:order} \\[2mm]
       &\scr_E(c) - \textstyle\sum_{v \in V_{c \succ p}} \Delta w_{\pi_v^{-1}(c)}^v \cdot y_{(v, \pi_v^{-1}(c))} \leqs  \scr_E(p) + \textstyle\sum_{v \in V} \sum_{j \in [\pi_v^{-1}(p) - 1]} \Delta w^v_j \cdot y_{(v, j)}  ,\forall c \in C_{\bad} \label{eq:more-score} \\
      &\textstyle\sum_{v \in V} \sum_{j \in [\pi_v^{-1}(p) - 1]} \Delta w^v_j \cdot y_{(v, j)} \geqs \textstyle\sum_{v \in V} \sum_{j \in [\pi_v^{-1}(p) - 1]} \Delta w^v_j \cdot y^*_{(v, j)}  \label{eq:cond-num-shift}
    \end{align}
\end{framed} 
\caption{\label{fig:lp2}LP2 for the proof of
  Lemma~\ref{lem:ptas-borda-main}.  For each voter $v$, we have
  variables $y_{(v,1)}, \ldots, y_{(v,m)}$.  For an
  integral solution, Constraints~\eqref{eq:order} ensure that the
  variables specify a shift action that pushes $p$ at least as far as
  shift action $\bs^*$ does, Constraints~\eqref{eq:more-score} ensure that
  $p$'s score at least matches the scores of candidates in $C_\bad$,
  and Constraint~\eqref{eq:cond-num-shift} ensures that $p$'s score is
  higher than the scores of candidates not in $C_\bad$.}
\end{figure}

\begin{proof}
  We first describe the somewhat non-intuitive algorithm, then
  we explain its workings and argue why it produces the desired
  approximate solution.  Let $I = (E,p, \psi)$ be an instance of
  $\cR$-\textsc{Shift-Bribery}, where $E = (C,V)$ is an election, $p$
  is the preferred candidate, and $\psi = \{\psi_v\}_{v \in V(E)}$ is
  a collection of price functions. Further, let $\cR$ be a positional
  scoring rule specified via scoring vectors $(\bw^v)_{v \in V}$ (with
  one vector for each voter in $V$). Our algorithm proceeds as
  follows:

  \begin{enumerate}
  \item We solve linear program LP1 from Figure~\ref{fig:lp1}.  Let
    $(x^{\opt}_{(v, j)})_{v \in V, j \in [m]}$ be the computed optimal
    solution found for this program. Note that the value of the
    optimization goal for LP1 is at most $\opt(I)$, because this would
    be the cost of an optimal integral solution.

  \item For every $v \in V, j \in [m]$, we let
    $y^*_{(v, j)} = \min\{1, (1 + \varepsilon)x^{\opt}_{(v, j)}\}$,
    and we let $j_v \in [m]$ be the smallest index such that
    $y^*_{(v, j_v)} = 1$. Intuitively, the shift action $\bs^*$ that
    for each voter $v$ shifts $p$ to position $j_v$ is our ``first
    order approximation'' of the shift action that we will eventually
    produce; its cost is at most $(1+\varepsilon)\opt(I)$ but after
    applying it, $p$'s score might still be lower than that of some of
    the candidates.
    Formally, we define set $C_\bad$ to contain
    all the candidates $c$ such that:
    \begin{align}
      \scr_E(c) - \sum_{v \in V_{c \succ p}} \Delta w^v_{\pi^{-1}_v(c)} \cdot \ind[\pi^{-1}_v(c) \geqs j_v] > \scr_E(p) + \bigg(\sum_{v \in V} \sum_{j \in [\pi_v^{-1}(p) - 1]} \Delta w^v_j \cdot y^*_{(v, j)} \bigg). \label{eq:easy-score}
    \end{align}
    On the left-hand side of the above equation, candidate $c$ loses
    as many points as indicated by shift action $\bs^*$, but the score
    of $p$, on the right-hand side, is computed with respect to the
    possibly fractional values $y^*_{(v,j)}$.


  \item We solve linear program LP2 from Figure~\ref{fig:lp2} for an
    optimal, basic solution
    $(y^{\opt}_{(v, j)})_{v \in V, j \in [m]}$. We output shift
    action $\bs$ that corresponds to
    $(\lceil y^{\opt}_{(v, j)}\rceil)_{v \in V, j \in [m]}$.
  \end{enumerate}

  The algorithm certainly runs in polynomial time. Let us now explain
  why the shift action that it outputs indeed ensures that $p$ is a
  winner. Foremost, due to Constraints~\eqref{eq:order}, shift
  action $\bs$ (weakly) dominates $\bs^*$ (i.e., for each voter it
  shifts $p$ at least as far as $\bs^*$ does). Thus, after applying
  $\bs$, each opponent of $p$ who is not in $C_\bad$ has score at most
  as high as in the left-hand side of Equation~\eqref{eq:easy-score}.
  Constraint~\eqref{eq:cond-num-shift} ensures that $p$ obtains at
  least as high a score as on the right-hand side of
  Equation~\eqref{eq:easy-score} and, thus, $p$ does not lose against any
  candidate not in $C_\bad$.  On the other hand,
  Constraints~\eqref{eq:more-score} ensure that $p$ does not lose against
  anyone in $C_\bad$.

  It remains to argue that $\cost_I(\bs)$ is at most as required in
  the lemma.  To this end, we first claim that $|C_\bad| < 1/\varepsilon$.

  \begin{claim} \label{claim:bound-unsat}
    $|C_{\bad}| < 1/\varepsilon$.
  \end{claim}

  \begin{subproof}
    Consider any $c \in C_{\bad}$. From the definition of $C_{\bad}$, we have:
    \begin{align*}
      \scr_E(p) + \left(\textstyle \sum_{v \in V} \sum_{j \in [\pi_v^{-1}(p) - 1]} \Delta w^v_j \cdot y^*_{(v, j)} \right) < 
      &\scr_E(c) - \left(\textstyle \sum_{v \in V_{c \succ p}} \Delta w^v_{\pi^{-1}_v(c)} \cdot \ind[\pi^{-1}_v(c) \geqs j_v]\right).
    \end{align*}
    On the other hand, from~\eqref{eq:more-score-1}, we have:
    \begin{align*}
      \scr_E(p) + \left(\textstyle\sum_{v \in V} \sum_{j \in [\pi_v^{-1}(p) - 1]} \Delta w^v_j \cdot x^{\opt}_{(v, j)}\right) \geqs
      &\scr_E(c) - \left(\textstyle\sum_{v \in V_{c \succ p}} \Delta w_{\pi_v^{-1}(c)}^v \cdot x^{\opt}_{(v, \pi_v^{-1}(c))}\right).
    \end{align*}
    Subtracting the latter inequality from the former gives:
    \begin{align} \label{eq:diff} 
      \sum_{v \in V}  \sum_{j \in [\pi_v^{-1}(p) - 1]} \Delta w^v_j \cdot (y^*_{(v, j)} - x^{\opt}_{(v, j)})
      &< \textstyle\sum_{v \in V_{c \succ p}} \Delta w_{\pi_v^{-1}(c)}^v \cdot (x^{\opt}_{(v, \pi_v^{-1}(c))} - \ind[\pi^{-1}_v(c) \geqs j_v]).
    \end{align}
    We lower bound the left-hand side of the above as follows (the
    inequality follows because we sum fewer non-negative terms, and
    the equality follows from the fact that for $j < j_v$, we have
    $y^*_{(v, j)} = (1+\varepsilon)x^{\opt}_{(v, j)}$):
    \begin{align*}
      \sum_{v \in V} \sum_{j \in [\pi_v^{-1}(p) - 1]} \Delta w^v_j \cdot (y^*_{(v, j)} - x^{\opt}_{(v, j)})
      &\geqs \textstyle\sum_{v \in V} \sum_{j \in [j_v - 1]} \Delta w^v_j \cdot (y^*_{(v, j)} - x^{\opt}_{(v, j)})\\
      &= \varepsilon \textstyle\sum_{v \in V} \sum_{j \in [j_v - 1]} \Delta w^v_j \cdot x^{\opt}_{(v, j)}.
    \end{align*}
    We upper bound the right-hand side of~\eqref{eq:diff} as follows,
    by dropping some possibly negative terms and adding some
    non-negative ones:
    \begin{align*}
      \sum_{v \in V_{c \succ p}} \Delta w_{\pi_v^{-1}(c)}^v \cdot (x^{\opt}_{(v, \pi_v^{-1}(c))} - \ind[\pi^{-1}_v(c) \geqs j_v])
      \leqs \textstyle\sum_{v \in V} \ind[\pi^{-1}_v(c) < j_v] \cdot \Delta w_{\pi_v^{-1}(c)}^v \cdot x^{\opt}_{(v, \pi_v^{-1}(c))}.
    \end{align*}
    Plugging the above two inequalities back to~\eqref{eq:diff}, we
    get:
    \begin{align*}
      \varepsilon &\textstyle\sum_{v \in V} \textstyle\sum_{j \in [j_v - 1]} \Delta w^v_j \cdot x^{\opt}_{(v, j)} < \textstyle\sum_{v \in V} \ind[\pi^{-1}_v(c) < j_v] \cdot \Delta w_{\pi_v^{-1}(c)}^v \cdot x^{\opt}_{(v, \pi_v^{-1}(c))}.
    \end{align*}
    Summing this over all candidates $c \in C_{\bad}$, we have the
    following (the final equality follows by replacing the summation
    over candidates $c$ with a summation over positions $j$, $j < j_v$
    where these candidates are ranked):
    \begin{align*}
      |C_{\bad}| \cdot \varepsilon \textstyle\sum_{v \in V} \sum_{j \in [j_v - 1]} \Delta w^v_j \cdot x^{\opt}_{(v, j)}
                 &< \textstyle\sum_{c \in C_{\bad}} \sum_{v \in V} \ind[\pi^{-1}_v(c) < j_v] \cdot \Delta w_{\pi_v^{-1}(c)}^v \cdot x^{\opt}_{(v, \pi_v^{-1}(c))} \\
                 &\leqs \textstyle\sum_{c \in C} \sum_{v \in V} \ind[\pi^{-1}_v(c) < j_v] \cdot \Delta w_{\pi_v^{-1}(c)}^v \cdot x^{\opt}_{(v, \pi_v^{-1}(c))} \\
                 &= \textstyle\sum_{v \in V} \sum_{j \in [j_v - 1]} \Delta w_j^v \cdot x^{\opt}_{(v, j)}.
    \end{align*}
    In the above, the expression on the left-most side is of the form
    $|C_\bad|\cdot\varepsilon$ times the expression on the right-most
    side, so we get $|C_{\bad}| < 1/\varepsilon$, as desired.
  \end{subproof}

  We now proceed to bound $\cost_I(\bs)$. First, observe that
  an optimal integral solution for LP1 has cost $\opt(I)$ and, thus,
  for our optimal, but perhaps non-integral, solution 
  $(x_{(v, j)})_{v \in V, j \in [m]}$ we have:
  \begin{align} \label{eq:LP1-v-OPT} 
    \textstyle\sum_{v \in V} \sum_{j \in
    [\pi_v^{-1}(p) - 1]}\Delta\psi_{v}(\pi_v^{-1}(p) - j) 
    x^{\opt}_{(v, j)} \leqs \opt(I).
  \end{align}

  For each voter $v \in V$, we say that $v$ is \emph{integral} if
  $y^{\opt}_{(v, j)} \in \{0, 1\}$ for all $j \in [m]$. Let $\Vi$
  denote the set of all integral voters.  Recall that
  $(y^{\opt}_{(v, j)})_{v \in V, j \in [m]}$ is a basic solution for
  LP2, meaning that exactly $m n$ linearly independent inequalities
  must be tight (because we have $mn$ variables).\footnote{We stress
    here that the inequalities must be linearly independent because in
    Constraint~\eqref{eq:order} we have equalities, each defined by
    two linearly dependent inequalities; satisfying such an equality
    ``counts'' as only one tight inequality for a basic solution.}
  For each non-integral voter $v \notin \Vi$, only at most $m - 1$
  linearly independent inequalities in~\eqref{eq:order} are tight.
  However, there are only
  $1 + |C_{\bad}| < 1 + 1/\varepsilon$ inequalities of the
  form~\eqref{eq:more-score} and~\eqref{eq:cond-num-shift}. From this,
  we can conclude that less than $1 + 1/\varepsilon$ voters are not
  integral, i.e.:
  \begin{align} \label{eq:vnonint} |V \setminus \Vi| < 1 +
    1/\varepsilon.
  \end{align}
  As a result, we have that $\cost_I(\bs)$ equals:
  \begin{align*}
    & \sum_{v \in V} \sum_{j \in [\pi_v^{-1}(p) - 1]} \Delta\psi_{v}(\pi_v^{-1}(p) - j)\cdot \lceil y^{\opt}_{(v, j)} \rceil \\
                 &= \textstyle\sum_{v \in \Vi} \sum_{j \in [\pi_v^{-1}(p) - 1]} \Delta\psi_{v}(\pi_v^{-1}(p) - j)\cdot y^{\opt}_{(v, j)} 
                 + \textstyle\sum_{v \in V \setminus \Vi} \sum_{j \in [\pi_v^{-1}(p) - 1]} \Delta\psi_{v}(\pi_v^{-1}(p) - j)\cdot \lceil y^{\opt}_{(v, j)} \rceil. 
  \end{align*}
  Now, observe that the first summation on the right hand side is
  upper bounded by the optimum of LP2. Note that
  $(y^*_{(v, j)})_{v \in V, j \in [m]}$ is a solution to LP2. Hence,
  we have:

  \begin{align*}
    \cost_I(\bs) \leqs \hspace{33pt} \textstyle\sum_{v \in V} \hspace{17pt} \sum_{j \in [\pi_v^{-1}(p) - 1]} &\Delta\psi_{v}(\pi_v^{-1}(p) - j) \cdot y^*_{(v, j)} \\
                 + \textstyle\sum_{v \in V \setminus \Vi} \sum_{j \in [\pi_v^{-1}(p) - 1]} &\Delta\psi_{v}(\pi_v^{-1}(p) - j)\cdot \lceil y^{\opt}_{(v, j)} \rceil \\
                 \leqs (1 + \varepsilon) \textstyle\sum_{v \in V} \hspace{17pt} \textstyle\sum_{j \in [\pi_v^{-1}(p) - 1]} &\Delta\psi_{v}(\pi_v^{-1}(p) - j)\cdot x^{\opt}_{(v, j)} \\
                 + \textstyle\sum_{v \in V \setminus \Vi} \sum_{j \in [\pi_v^{-1}(p) - 1]} &\Delta\psi_{v}(\pi_v^{-1}(p) - j)\cdot \lceil y^{\opt}_{(v, j)} \rceil \\
                 \hspace{0pt} \overset{\eqref{eq:LP1-v-OPT}}{\leqs} \textstyle (1 + \varepsilon)\opt(I)
                 + \textstyle \sum_{v \in V \setminus \Vi} \sum_{j \in [\pi_v^{-1}(p) - 1]} &\Delta\psi_{v}(\pi_v^{-1}(p) - j)\cdot \lceil y^{\opt}_{(v, j)} \rceil.
  \end{align*}

  Now, observe that for each $v \in V$, if
  $\sum_{j \in [\pi_v^{-1}(p) - 1]} \Delta\psi_{v}(\pi_v^{-1}(p) -
  j)\cdot \lceil y^{\opt}_{(v, j)} \rceil$
  is $\infty$, then $\opt(I)$ must also be $\infty$ and the inequality
  we try to prove is trivially true. Hence, we may assume that
  $\sum_{j \in [\pi_v^{-1}(p) - 1]} \Delta\psi_{v}(\pi_v^{-1}(p) -
  j)\cdot \lceil y^{\opt}_{(v, j)} \rceil$
  is finite; in this case, this quantity is bounded by
  $\psi^{\max}(I)$. As a result, we can further bound $\cost_I(\bs)$
  by
  \begin{align*}
    \cost_I(\bs) &\leqs (1 + \varepsilon)\opt(I) + |V \setminus \Vi| \cdot \psi^{\max}(I) \\
                 &\overset{\eqref{eq:vnonint}}{\leqs} (1 + \varepsilon)\opt(I) + (1 + 1/\varepsilon) \psi^{\max}(I),
  \end{align*}
  which concludes our proof.
\end{proof}

\subsection{The Final PTAS}
Now we use the approximation algorithm with additive error to derive
an approximation algorithm with a purely multiplicative ratio. The
intuition behind this process is simple:
Since the algorithm from Lemma~\ref{lem:ptas-borda-main} works
well when $\psi^{\max}(I)$ is small, we will
first ``preprocess'' our instance so that $\psi^{\max}(I)$ is much
smaller than $\opt(I)$. To do so, note that if we consider an optimal
shift action $\bs^{\opt}$, then there are only a few voters $v$ such
that $\psi_v(s^{\opt}_v)$ is large; specifically, for every
$\delta > 0$, only at most $1/\delta$ voters have
$\psi_v(s^{\opt}_v) \geqs \delta \opt(I)$. This means that if we guess
such voters and the numbers of unit shifts that we apply to
them, then we can reduce the instance $I$ to another instance $I'$,
where $\psi^{\max}(I')$ is bounded by $\delta \opt(I)$. We then run
the algorithm from Lemma~\ref{lem:ptas-borda-main} on  $I'$. 
By selecting $\delta = O(\varepsilon^2)$, the additive
error becomes
$O((\delta / \varepsilon)\opt(I)) = O(\varepsilon \opt(I))$ as
intended. 

\begin{theorem} \label{thm:ptas-borda-main} There is an algorithm that
  given $\varepsilon > 0$ and an instance $I$ of
  $\cR$-\textsc{Shift-Bribery}, where $\cR$ is a given positional
  scoring rule with a possibly different scoring vector for each
  voter, outputs a successful shift action for $I$ of
  cost at most $(1 + \varepsilon)\opt(I)$ and runs in time
  $|I|^{O(1/\varepsilon^2)}$.
\end{theorem}

\begin{proof}[Proof of Theorem~\ref{thm:ptas-borda-main}]
  We are given an instance $I = (E,p, \psi)$ of
  $\cR$-\textsc{Shift-Bribery}, where $E = (C,V)$ is an election, $p$
  is the preferred candidate, $\psi = \{\psi_v\}_{v \in V})$ is a
  collection of price functions, and $\cR$ is a positional scoring
  rule (possibly with a different scoring vector for each voter in
  $V$). We also have $\varepsilon > 0$.  Our algorithm works as
  follows:
  \begin{enumerate}
  \item Let $\delta = \varepsilon^2/8$ and
    $q = \lceil 1/\delta \rceil$.
  \item For every subset $S \subseteq V$ of $q$ voters and every
    possible shift action $(s_v)_{v \in S} \in (\{0\} \cup [m])^S$ for
    these voters, we execute the following steps:
    \begin{enumerate}
    \item Compute $b = \min_{v \in S} \psi_v(s_v)$, the minimum cost
      among all shifts $(s_v)_{v \in S}$.
    \item Create new price functions $\{\psi'_v\}_{v \in V}$ as
      follows.  For all $v \in S$, let $\psi'_v(j)$ be $0$ for
      $j \leqs s_v$ and let it be $\infty$ for the remaining values of
      $j$. For all $v \notin S$, let $\psi'_v(j) = \psi_v(j)$ whenever
      $\psi_v(j) \leqs b$, and let $\psi'_v(j) = \infty$ otherwise.
    \item
      Run the algorithm from Lemma~\ref{lem:ptas-borda-main} on the
      instance $I' = (E, \{\psi'_v\}_{v \in V})$ and
      $\varepsilon/2$. \label{step:run-b}
    \end{enumerate}
  \item Output the shift action with minimum cost among all shift
    actions produced by Step~\ref{step:run-b}.
  \end{enumerate}
  First, notice that the algorithm runs in time
  $|I|^{O(q)} = |I|^{O(1/\varepsilon^2)}$ and that the algorithm
  outputs a successful shift action (i.e., $p$ is indeed a winner
  after the shifts).

  To see that the cost of the output solution is at most
  $(1 + \varepsilon) \opt(I)$, let us consider an optimal shift action
  $\bs^{\opt} = (s^{\opt}_v)_{v \in V}$ of $I$, i.e., we have
  $\cost_I(\bs) = \opt(I)$. Now, let $S$ denote the set of $q$ voters
  whose shifts are the most expensive, i.e., $S \subseteq V$ is the
  set of size $q$ such that $\psi(s^{\opt}_v) \geqs \psi(s^{\opt}_u)$
  for all $v \in S, u \notin S$. This set $S$ and the shifts
  $(s^{\opt}_v)_{v \in S}$ are considered by the above algorithm. In
  this case, observe that:
  \begin{align} \label{eq:bound-b} b = \min_{v \in S}
    \psi_v(s^{\opt}_v) \leqs \frac{1}{q} \sum_{v \in S}
    \psi_v(s^{\opt}_v) \leqs \delta \opt(I).
  \end{align}
  Observe also that $\psi^{\max}(I') \leqs b$. Thus,
  Step~\ref{step:run-b} gives a shift action
  $\bs^* = (s^*_v)_{v \in V}$ such that:
  \begin{align} \label{eq:bound-cost-i-prime} \cost_{I'}(\bs^*)  \leqs
    (1 + \varepsilon/2)\opt(I') + (1 + 2/\varepsilon)b 
     \overset{\eqref{eq:bound-b}}{\leqs} (1 + \varepsilon/2)\opt(I') +
    \varepsilon\opt(I)/2.
  \end{align}
  Next, notice that:
  \begin{align} \label{eq:opt-i-to-iprime} \opt(I') \leqs
    \cost_{I'}(\bs^{\opt}) = \opt(I) - \sum_{v \in S}
    \psi_v(s^{\opt}_v)
  \end{align}
  and:
  \begin{align} \label{eq:i-to-iprime} \cost_{I}(\bs^*) \leqs
    \cost_{I'}(\bs^*) + \sum_{v \in S} \psi_v(s^{\opt}_v).
  \end{align}
  Combining \eqref{eq:bound-cost-i-prime}, \eqref{eq:opt-i-to-iprime}
  and \eqref{eq:i-to-iprime}, we obtain the following bound on
  $\cost_{I}(\bs^*)$:
  \begin{align*}
    \cost_{I}(\bs^*) &\overset{\eqref{eq:i-to-iprime}}{\leqs} \cost_{I'}(\bs^*) + \sum_{v \in S} \psi_v(s^{\opt}_v) \\
                     &\overset{\eqref{eq:bound-cost-i-prime}}{\leqs} (1 + \varepsilon/2)\opt(I') + \varepsilon\opt(I)/2 + \sum_{v \in S} \psi_v(s^{\opt}_v) \\
                     &\overset{\eqref{eq:opt-i-to-iprime}}{\leqs} (1 + \varepsilon)\opt(I) - (\varepsilon/2) \sum_{v \in S} \psi_v(s^{\opt}_v) \\
                     &\leqs (1 + \varepsilon) \opt(I).
  \end{align*}
  Hence, the output shift action has cost at most
  $(1 + \varepsilon)\opt(I)$ as desired.
\end{proof}

We remark that as a corollary of Theorem~\ref{thm:ptas-borda-main}, we
also get a PTAS for Borda-\textsc{Shift-Bribery} for arbitrary prices.

\section{Copeland}\label{sec:copeland}

For the case of Copeland$^\alpha$ family of rules, we show that the
\textsc{Shift-Bribery} problem is hard to approximate even for the
unit prices and for the all-or-nothing prices.  Specifically, we show
that an approximation algorithm for the
$\copeland$-\textsc{Shift-Bribery} implies an approximation algorithm
for the \textsc{Densest-$k$-Subgraph} problem, which is 
believed to be hard to approximate~\cite{BCVGZ12}.
Below we describe our results generally, and then we provide their
proofs for respective types of price functions in the following
sections.

\begin{definition}\label{def:dks}
  In the \textsc{Densest-$k$-Subgraph} (D$k$S) problem, we are given
  an undirected graph $G = (V_G, E_G)$ and a positive integer $k$, and
  the goal is to output a $k$-vertex subgraph of $G$ with as many
  edges as possible.
\end{definition}

\begin{theorem} \label{thm:copeland} Let $\tau$ be an arbitrary
  non-decreasing function. If there is a polynomial time
  $\tau(|I|)$-approximation algorithm for
  $\copeland$-\textsc{Shift-Bribery} for some $\alpha \in [0, 1]$, for
  the case of unit prices or all-or-nothing prices, then there is a
  polynomial time $O(\tau(|V_G|^{O(1)})^2)$-approximation algorithm for
  the \textsc{D$k$S} problem.
\end{theorem}

Although hardness of approximation of D$k$S within up to polynomial
factor is not known, inapproximability up to almost polynomial factor
is known assuming the \emph{exponential time hypothesis} (ETH) and its
gap version (Gap-ETH).\footnote{ETH~\cite{IP01,IPZ01} states that there
  is no subexponential time algorithm that solves
  3SAT. Gap-ETH~\cite{D16,MR17} states that no subexponential
  time algorithm can distinguish between a satisfiable 3CNF
  formula and one which is only $(1 - \delta)$-satisfiable for some
  absolute constant $\delta > 0$.} Specifically,
Manurangsi~\cite{Man17} has shown that under the ETH assumption
(the Gap-ETH assumption, respectively), \textsc{Densest-$k$-Subgraph}
is hard to approximate to within a factor of
$n^{1/\poly(\log \log n)}$ ($n^{o(1)}$, respectively). Together with
Theorem~\ref{thm:copeland}, this implies the following corollary.

\begin{corollary} \label{cor:copeland} Assuming ETH, for some constant
  $c > 0$ there is no polynomial-time
  $|I|^{1/(\log \log |I|)^c}$-approximation algorithm for
  $\copeland$-\textsc{Shift-Bribery} for any $\alpha > 0$, even for
  unit prices or all-or-nothing prices. Moreover, assuming Gap-ETH,
  the inapproximability ratio can be improved to $|I|^{f(|I|)}$ for
  any function $f = o(1)$.
\end{corollary}

For the parameterization by the number of unit shifts, assuming Gap-ETH
implies that there is no $\FPT$ approximation scheme for the problem
even for the case of unit prices.

\begin{theorem} \label{thm:copeland-unit-fpt} Assuming Gap-ETH, for
  every $\alpha \in [0, 1]$, every $\varepsilon > 0$, and every
  computable function $T$, there is no algorithm that given a
  $\copeland$-\textsc{Shift-Bribery} instance $I$
  with unit prices, runs in time
  $T(\opt(I)) \cdot \poly(|I|)$
  and outputs a successful shift action with at most
  $(2 - \varepsilon) \opt(I)$ unit shifts.
\end{theorem}

We are not aware of a constant factor $\FPT$ approximation
algorithms for the problem and it is possible that the factor $2$
above can be improved to larger constants, or even beyond a
constant. This remains an interesting open question.

Another parameter that has been considered in the literature is the
number of affected voters. For this parameter, the exact version is
known to be $\W[2]$-hard (for both unit prices and all-or-nothing
prices) by a reduction from $k$-Set Cover~\cite{BCFNN16}. Due to the
recent developments in parameterized inapproximability of $k$-Set
Cover~\cite{ChenL16,KLM18}, we can modify the $\W[2]$-hardness proof
to yield strong parameterized inapproximability results for
$\copeland$-\textsc{Shift-Bribery} with this parameter.

\begin{theorem} \label{thm:copeland-unit-fpt-voter} For every $\alpha
  \in [0, 1]$, every $\varepsilon > 0$,
  $\copeland$-\textsc{Shift-Bribery} parameterized by the number of
  affected voters is $\W[1]$-hard to approximate to within any
  constant factor, even for unit prices or $(1,
  \infty)$-all-or-nothing prices.
\end{theorem}

Notice that the hardness in Theorem~\ref{thm:copeland-unit-fpt-voter} is
stronger than Theorem~\ref{thm:copeland-unit-fpt} both in terms of
the approximation ratio and that it requires weaker assumption.
We also remark here that, for $(1, \infty)$-all-or-nothing prices,
the optimum is exactly equal to the number of affected voters.
Hence, the hardness above carries over to the ``budget'' parameter as well.
To summarize, our results implies inapproximability for
essentially all parameters left open by~\cite{BCFNN16}.

\subsection{Additional Notation and Tools} \label{sec:add-notation}

We need some additional notation in the following text. Let $C$ be a
set of candidates. We write $\left< C \right>$ to denote an arbitrary
(but fixed) preference order over $C$, and for each $A \subseteq C$,
by $\left< A \right>$ we mean $\left<C\right>$ restricted to the
candidates from $A$. For example, for $C = \{a,b,c,d,e\}$ and
$A = \{a,b,c\}$ we may write $d \succ \left< A \right> \succ e$ to
denote preference order where $d$ is ranked ahead of all the members
of $A$, and $e$ is below all members of $A$ (and below $d$, of
course). When we use this notation to specify a preference order,
typically the ranking of candidates within $\left< A \right>$ is
irrelevant for the argument.  For each $A \subseteq C$, we write
$\overleftarrow{\left<A\right>}$
to denote the reverse of the order~$\left<A\right>$.

In Definition~\ref{def:dks} we have defined the optimization variant
of the \textsc{Densest-$k$-Subgraph} problem. In the decision variant,
an instance consists of a graph $G$ and two positive integers, $k$ and
$t$, and we ask if it is possible to select $k$ vertices that jointly
induce a graph with at least $t$ edges.

Given a graph $G = (V_G,E_G)$ and a set of vertices $U \subseteq V_G$,
by $E_G[U]$ we mean the set of edges from $E_G$ that have both
endpoints in $U$.

Often, when
constructing hard instances for $\copeland$-\textsc{Shift-Bribery}, we
want to create additional voters and candidates so that the candidates
of our interest have certain scores and certain outcomes of pairwise
elections. The following proposition is especially useful for this
purpose; it is extracted from the proof of Theorem 4 of~\cite{BCFNN16}
with slight modifications. We provide a proof for the sake of
completeness.

\begin{proposition}[Bredereck et
  al.~\cite{BCFNN16}] \label{prop:dummy} Let $C$ be a candidate
  set that can be partitioned into $A \cup B \cup \{p, d\}$, where
  $|A| + |B|$ is an odd integer, and let $a, b$ be non-negative
  integers such that $a \leqs |B|$. Then, there exists an election
  $E = (C, V, \{\succ_v\}_{v \in V})$ with
  $|V| = 2|A| + 2|B| + 2b + 5$ such that:
  \begin{enumerate}
  \item $p$ loses pairwise elections against every candidate in $A$ by
    $2b + 1$ votes,
  \item $p$ receives $|B| - a + 1$ points (with respect to the $\copeland$
    rule),
  \item $d$ receives $|B|$ points, and
  \item every candidate in $A \cup B$ receives at most
    $(|A| + |B| + 3)/ 2$ points.
  \end{enumerate}
  (Note here that, since $|C|$ is an odd integer, there are no ties in
  pairwise elections.)
\end{proposition}
\begin{proof}[Proof of Proposition~\ref{prop:dummy}]
  The construction here is a slight adaptation of the construction
  given in the proof of Theorem 4 by Bredereck et
  al.~\cite{BCFNN16}. For convenience, let us define some
  additional notation (also due to Bredereck et al.~\cite{BCFNN16}):
  For any set $T = \{c_1, \dots, c_{|T|}\}$ of an odd number of
  candidates and any $c_i \in T$, we define $\hso(T, c_i)$ and
  $\hst(T, c_i)$ as the preferences orders:
  \begin{align*}
    c_i \succ \left<T_i\right> \succ \left<T\setminus T_i\right>  
    && \text{and} &&
    \overleftarrow{\left<T\setminus T_i\right>} \succ c_i \succ \overleftarrow{\left<T_i\right>}
  \end{align*}
  respectively, where
  $T_i = \{c_{i + 1}, \dots, c_{i + (|T| - 1)/2}\}$; here we use the
  convention that $c_{|T| + j} = c_j$. The key property here is that
  if there is one voter with a preference order that includes
  $\hso(T, c_i)$ and one with a preference order that includes
  $\hst(T, c_i)$, then $c_i$ wins pairwise election against exactly
  half of the candidates in $T \setminus \{c_i\}$ (i.e. those in
  $T_i$) and ties with the rest. All other pairwise elections not
  involving $c_i$ result in ties.

  Let $S \subseteq B$ be any subsets of $B$ of size $a$. For our
  election $E$, we create the following voters:
  \begin{enumerate}
  \item One voter with preference order
    $\left<A \cup B\right> \succ p \succ d$.
  \item $b$ pairs of voters with preference orders
    $\left<A\right> \succ p \succ \left<B\right> \succ d$ and
    $d \succ \overleftarrow{\left<B\right>} \succ
    \overleftarrow{\left<A\right>} \succ p$.
  \item Two voters with preference orders
    $d \succ \left<B\right> \succ \left<A\right> \succ p$ and
    $p \succ \overleftarrow{\left<A\right>} \succ d \succ
    \overleftarrow{\left<B\right>}$.
  \item Two voters with preference orders
    $p \succ \left<B \setminus S\right> \succ \left<S\right> \succ
    \left<A\right> \succ d$
    and
    $d \succ \overleftarrow{\left<A\right>} \succ
    \overleftarrow{\left<S\right>} \succ p \succ
    \overleftarrow{\left<B \setminus S\right>}$.
  \item For every candidate $c \in A \cup B$, two voters with
    preference orders $\hso(A \cup B, c) \succ p \succ d$ and
    $d \succ p \succ \hst(A \cup B, c)$.
  \end{enumerate}
  It is simple to verify that (i) $p$ loses pairwise elections against
  every candidate in $A$ by $2b + 1$ votes, (ii) $p$ wins pairwise
  elections exactly against the candidates in the set
  $\{d\} \cup (B \setminus S)$, (iii) $d$ wins pairwise elections
  exactly with the candidates from $B$, and (iv) each candidate in
  $A \cup B$ wins pairwise elections against exactly half of the other
  candidates in $A \cup B$. These four features of $E$ indeed imply
  the claimed properties in the statement of the proposition.
\end{proof}

\subsection{All-or-Nothing Prices} \label{sec:copeland-01}

The goal of this section is to prove Theorem~\ref{thm:copeland} for
the case of all-or-nothing prices. To this end, we focus on the next
lemma, from which the desired result follows.
The proof is
similar to the $\W[1]$-hardness proof of
$\copeland$-\textsc{Shift-Bribery} of Bredereck et
al.~\cite{BCFNN16} except that the roles of edges and vertices
are reversed.

\begin{lemma} \label{lem:copeland-01} For each $\alpha \in [0, 1]$,
  there exists a reduction that takes in a
  \textsc{Densest-$k$-Subgraph} instance $(G, k, t)$ and outputs an
  instance $I$
  of $\copeland$-\textsc{Shift-Bribery}
  with $(1,\infty)$-all-or-nothing prices of width $|V_G|$
  (where $V_G$ is the set of $G$'s vertices)
  such that the following holds for every $0 < \delta \leqs 1$:
  \begin{enumerate}
  \item (Completeness) If there exists a $k$-vertex subgraph of $G$
    with $t$ edges, then $\opt(I) \leqs k$.
  \item (Soundness) If every $k$-vertex subgraph of $G$ contains fewer
    than $\delta t$ edges, then $\opt(I) > (k - 1)/\sqrt{\delta}$.
  \end{enumerate}
  Moreover, the reduction runs in $\poly(|V_G|, t)$ time.
\end{lemma}


\begin{proof}
  Consider an instance $(G,k,t)$ of the decision variant of the
  \textsc{Densest-$k$-Subgraph} problem, where $G = (V_G,E_G)$ is a
  graph and $k$, $t$ are two integers.
  Our reduction forms an instance $I = (E,\psi,p)$ of the
  $\copeland$-\textsc{Shift-Bribery} problem, with election
  $E = (C,V, {\{\succ_v\}_{v \in V}})$..
  We set $C = E_G \cup D \cup \{p,d\}$, where $D$ is a set of
  $|E_G| + 5$ dummy candidates, $p$ is the preferred candidate, and
  $d$ will be the unique winner of the election prior to shifting $p$.
  The voters are constructed as follows (for a vertex $u$, by
  $\Gamma_G(u)$ we mean the set of edges incident to it):
  \begin{enumerate}
  \item For each vertex $u \in V_G$, we create two voters, $v_u$ and
    $v'_u$, such that the preference order of $v_u$ is:
    \begin{align*}
      \big<\Gamma_G(u)\big> \succ_{v_u} p \succ_{v_u} \left<C \setminus \big(\{p\} \cup \Gamma_G(u)\big)\right>.
    \end{align*}
    and the preference order of $v'_u$ is its reverse.  The cost of
    each of the possible shifts for $v_u$ is one, i.e.,
    $\psi_{v_u}(1) = \cdots = \psi_{v_u}(|\Gamma_G(u)|) = 1$, whereas
    the cost of each of the possible shifts for $v_u'$ is infinity,
    i.e.,
    $\psi_{v'_u}(1) = \cdots = \psi_{v'_u}(|C| - 1 - |\Gamma(u)|) =
    \infty$.
  \item We invoke Proposition~\ref{prop:dummy} with $A = E_G, B = D, b = 1$
    and $a = t + 1$ and create polynomially many additional voters so
    that $p$ loses pairwise elections against every candidate in $E_G$
    by $3$ votes, $\scr_E(p) = |D| - t, \scr_E(d) = |D|$ and, for all
    $c \in E_G \cup D$, $\scr_E(c) \leqs (|E_G| + |D| + 3)/2 < |D|$.
    Note that it is possible to obtain these scores via
    Proposition~\ref{prop:dummy} because for each vertex $u$, the preference
    orders of $v_u$ and $v'_u$ cancel each other out.
    For these voters, let the price of all the possible shifts be infinity.
  \end{enumerate}

  (Completeness) Suppose that there exists a set $U \subseteq V_G$ of
  $k$ vertices that induces a graph with at least $t$ edges.  Consider
  a shift action where for each $u \in V$ we shift $p$ to the first
  position in the preference order of $v_u$. Clearly, this shift
  action is of cost $k$. Moreover, for each edge
  $e = \{u_1, u_2\} \in E_G[U]$, the shifts switch the ranks of $e$
  and $p$ in two preference orders, those of $v_{u_1}$ and
  $v_{u_2}$. Since in the original election $p$ was losing the
  pairwise election against $e$ by 3 votes, after the shifts $p$ wins
  this pairwise election. As a result, after the shifts the score of
  $p$ is at least $|D| - t + |E_G[U]| \geqs |D| = \scr_E(d)$. Hence,
  $p$ is a winner of the election.

  (Soundness) 
  Suppose contrapositively
  that $\opt(I) \leqs (k - 1)/\sqrt{\delta}$. Suppose that
  $\bs = (s_v)_{v \in V}$ is a successful shift action with cost
  $\opt(I)$. Since we have all-or-nothing prices, we can assume that
  for every vote where $\bs$ shifts $p$, it shifts him or her to the top
  position in the vote.  Note that no shift action of finite cost can
  change the score of $d$. As a result, $p$ must end up having score
  at least $|D|$. Observe also that for each vote, no shift action of
  finite cost can change the relative order of $p$ and any of the
  candidates in $D$. This implies that after applying $\bs$, $p$ must
  win pairwise elections against at least $t$ candidates in $E_G$; let
  $Y \subseteq E_G$ denote the set of these edges.

  Now, let $U \subseteq V$ denote the set of all $u \in V_G$ such that
  at least one unit shift is applied to $v_u$ (i.e., $s_{v_u} > 0$).
  Note that $|U| = \opt(I) \leqs (k - 1)/\sqrt{\delta}$.  Furthermore,
  it is not hard to see that each edge $e = \{u_1, u_2\} \in E_G$
  loses pairwise election against $p$ if and only if $U$ contains both
  $u_1$ and $u_2$. In other words, we have
  $\big|E_G[U]\big| = |Y| \geqs t$. Let $U'$ be a random subset of $U$
  of size $k$ (or $U' = U$ when $|U|<k$).
  The expected number of edges induced by $U'$ is:
  \begin{align*}
    &\Ex\Big[\big|E_G[U']\big|\Big] = \sum_{u \in V}\sum_{\substack{v \in V\\\{u,v\} \in E_G}} \frac{1}{2}\Pr\left[u \in U' \wedge v \in U'\right] \\ 
    &= \frac{{|U| \choose k-2}}{{|U| \choose k}} \cdot \sum_{u \in V}\sum_{\substack{v \in V\\\{u,v\} \in E_G}} \frac{1}{2} \\
    &= \frac{k(k - 1)}{|U|(|U| - 1)} \cdot \big|E_G[U]\big| \geqs \frac{k(k - 1)}{(k - 1)^2/\delta} \cdot t \geqs \delta t.
  \end{align*}
  As a result, there exists a $k$-vertex subgraph of $G$ with
  $\delta t$ edges, which concludes our proof.
\end{proof}

\subsection{Unit Prices} \label{sec:copeland-unit}

The unit prices part of Theorem~\ref{thm:copeland} is established via
a reduction from the all-or-nothing price case; we state the reduction
in the general form below as it will be used again in the next
section.

\begin{lemma} \label{lem:copeland-red} For every $\alpha \in [0, 1]$
  and $B, B' \in \N$ such that $B \leqs B'$,
  there exists a $poly(|I|, B, B')$-time algorithm that takes in an
  instance $I$
  of $\copeland$-\textsc{Shift-Bribery}
  with $(1,
  \infty)$-all-or-nothing prices of width
  $b$ and produces an instance $I'$
  of $\copeland$-\textsc{Shift-Bribery} with unit prices, such that
  $\min\{B', B \cdot \opt(I)\} \leqs \opt(I') \leqs (B + b) \cdot
  \opt(I)$.
  Moreover, when $B > b \cdot \opt(I)$, any minimum cost successful
  shift action in $I'$ affects only $\opt(I)$ voters.
\end{lemma}

\begin{proof}
  Let the notation be as in the statement of the theorem and let
  $E = (C,V, \{\succ_v\}_{v \in V})$ be the election from
  instance~$I$.  We assume without loss of generality that $|V| \geqs 3$.
  We form an election $E' = (C', V, \{\succ'_{v}\}_{v \in V})$ with the
  same voters (but with modified preference orders), and with
  candidate set $C' = C \cup D$, where $D$ is the set of additional
  $(B + B') \cdot |V|$ ``filler candidates''. Let
  $\{D_{v}\}_{v \in V}$ be a collection of disjoint subsets of $D$
  such that $|D_{v}| = B$ if it costs one to shift $p$ in $v$'s
  preference list and $|D_{v}| = B'$ otherwise. For each voter
  $v \in V$, his or her preference order $\succ'_{v}$ in $E'$ is:
  \begin{align*}
    \pi_v(1) &\succ'_{v} \cdots \succ'_{v} \pi_v(\pi_v^{-1}(p) - 1) \succ'_{v} 
         \left<D_v\right> \succ'_{v}  p \\ 
       &\succ'_{v}  \pi_v(\pi_v^{-1}(p) + 1) \succ'_{v} \cdots \succ'_{v} \pi_v(|C|) \succ'_{v} \left<D \setminus D_v\right>.
  \end{align*}
  In other words, we add the filler candidates from $D_v$ right in
  front of $p$, and we put the rest of the filler candidates at the
  end of the list. Note that each filler candidate from $D$ loses
  pairwise election to each candidate from $C$.  We form an instance
  $I' = (E',p, \psi')$, where $\psi'$ are unit prices.

  It is not hard to check that if a shift action $\bs \in \N_0^{V}$ is
  successful for $I$, then the shift action $\bs' \in \N_0^{V}$
  defined by:
  \begin{align*}
    s'_{v} =
    \begin{cases}
      0 & \text{ if } s_v = 0, \\
      B + \pi_v^{-1}(p) - 1 & \text{ if } s_{v} > 0,
    \end{cases}
  \end{align*}
  is successful for $I'$. Moreover, since $I$ uses
  $(1, \infty)$-all-or-nothing price functions of width $b$, the cost
  incurred by $\bs'$ for each voter is at most $(B + b)$ times the
  original cost of that voter in $\bs$. As a result, we have
  $\opt(I') \leqs (B + b) \cdot \opt(I)$ as desired. 

  Next, we show that $\opt(I') \geqs \min\{B \cdot \opt(I), B'\}$.
  Suppose that $\opt(I') < B'$; we will show that
  $\opt(I') \geqs B \cdot \opt(I)$.  Consider any shift action
  $\bs' \in \N_0^V$ with cost $\opt(I')$ (with respect to the instance
  $I'$) that is successful for $I'$. We define $\bs \in \N_0^{V}$ by:
  \begin{align*}
    s_v = 
    \begin{cases}
      0 & \text{ if } s'_{v} \leqs |D_{v}|, \\
      s'_{v} - |D_{v}| & \text{ if } s'_{v} > |D_v|.
    \end{cases}
  \end{align*}
  It is simple to see that $\bs$ is successful for $I$. Moreover,
  since $\opt(I') < B'$ and for each voter $v \in V$ such that the
  cost of shifting $p$ in $v$ is $\infty$ (in terms of the instance
  $I$) we have that $|D_{v'}| = B'$, we see that $\bs$ does not
  contain infinity-priced shifts (in terms of $I$). Finally, for every
  voter $v$ such that $s_v > 0$, we have $s'_v > B$. This means that
  the cost of $\bs$ (in terms of $I$) is at most $\opt(I') /
  B$. Hence, $B \cdot \opt(I) \leqs \opt(I')$ as desired.

  Lastly, suppose that $B > b \cdot \opt(I)$.
  Consider any minimum cost successful shift action $\bs'$ of $I'$;
  we might assume without loss of generality that $s'_v$ is either zero
  or at least $B$. Otherwise, if $0 < s'_v < B$,
  we can change $s'_v$ to zero, which retains $p$ as a winner and also reduces
  the cost. From this, we have that the number of affected voters
  in $\bs'$ is at most
  \begin{align*}
  \frac{\opt(I')}{B} \leqs \frac{(B + b) \cdot \opt(I)}{B} < \opt(I) + 1.
  \end{align*}
  In other words, $\bs'$ affects at most $\opt(I)$ voters,
  which concludes our proof.
\end{proof}

We can now prove the desired inapproximability of $\copeland$ for unit
prices, by simply applying Lemmas~\ref{lem:copeland-01}
and~\ref{lem:copeland-red} together with appropriate values of $B$ and
$B'$. This idea is formalized below.

\begin{lemma} \label{lem:copeland-unit} For each $\alpha \in [0, 1]$,
  there exists a reduction that takes in a
  \textsc{Densest-$k$-Subgraph} instance $(G, k, t)$ and outputs an
  instance $I'$
  of $\copeland$-\textsc{Shift-Bribery}
  with unit prices such that the following holds for every $0 < \delta \leqs 1$
  ($V_G$ is the set of vertices for $G$):
\begin{enumerate}
\item (Completeness) If there exists a $k$-vertex subgraph of $G$ with
  $t$ edges, then $\opt(I') \leqs 2 |V_G| \cdot k$.
\item (Soundness) If every $k$-vertex subgraph of $G$ contains fewer
  than $\delta t$ edges, then
  $\opt(I') > \frac{|V_G| (k - 1)}{\sqrt{\delta}}$.
\end{enumerate}
Moreover, the reduction runs in $\poly(|V_G|, t)$ time.
\end{lemma}

\begin{proof}
  Let $(G,k,t)$ be our input instance of D$k$S,
  where $G  = (V_G,E_G)$ is a graph,  $k$
  is the number of vertices we can select, and $t$
  is the lower bound on the required number of edges in the graph
  induced by the selected vertices.  The reduction proceeds as
  follows. First, we apply Lemma~\ref{lem:copeland-01} to produce an
  instance $I$.
  Then, we invoke Lemma~\ref{lem:copeland-red} with $B = |V_G|$ and
  $B' = |V_G|^4 + 1$ to produce an instance $I'$,
  which we output. The reduction runs in polynomial time.

  (Completeness) If there exists a $k$-vertex subgraph of $G$ with $t$
  edges, then Lemma~\ref{lem:copeland-01} ensures that
  $\opt(I) \leqs k$. Moreover, Lemma~\ref{lem:copeland-red} then
  ensures that
  $\opt(I') \leqs (|V_G|-1 + |V_G|) \opt(I) \leqs 2|V_G| \cdot k$.

  (Soundness) If every $k$-vertex subgraph of $G$ contains fewer than
  $\delta t$ edges, then Lemma~\ref{lem:copeland-01} implies that
  $\opt(I) > (k - 1)/\sqrt{\delta}$. Then,
  Lemma~\ref{lem:copeland-red} ensures that
  $\opt(I') \geqs \min\{B', B \cdot \opt(I)\} > \min\{|V_G|^4, |V_G|
  (k - 1) / \sqrt{\delta}\} \geqs |V_G| (k - 1) / \sqrt{\delta}$.
  (The last inequality holds because we can assume w.l.o.g. that
  $\delta \geqs 1/t$ and that $t \leqs k^2 \leqs |V_G|^2$.)
\end{proof}

\subsection{FPT Inapproximability Results} \label{sec:copeland-fpt}

In this section we show that approximating
$\copeland$-\textsc{Shift-Bribery} is difficult even for FPT
algorithms, for the parameterizations by the number of unit shifts and
by the number of affected voters.

\subsubsection{Parameterization by the Number of Unit Shifts}

We first show $\FPT$ inapproximability of
$\copeland$-\textsc{Shift-Bribery} with unit prices, parameterized by
the number of unit shifts (Theorem~\ref{thm:copeland-unit-fpt}). To do
so, let us recall the following hardness result regarding
distinguishing a graph with a $k$-clique and one in which every
$k$-vertex subgraph is sparse, as proved by Chalermsook et
al.~\cite{CCKLMNT17} (which, in turn, relies heavily on the reduction
and the main lemma of Manurangsi~\cite{Man17}).

\begin{theorem}[Chalermsook et al.~\cite{CCKLMNT17}]\label{thm:dks-fpt-inapprox}
  Assuming Gap-ETH, for every computable function $T$ and every constant
  $\delta \in (0, 1)$, there is no algorithm that, given a graph
  $G = (V_G, E_G)$ and an integer $k$, can distinguish between the
  following two cases in time $T(k) \cdot |V_G|^{O(1)}$:
  \begin{enumerate}
  \item (Yes) There exists a $k$-clique in $G$.
  \item (No) Every $k$-vertex subgraph of $G$ contains fewer than
    $\delta \cdot \binom{k}{2}$ edges.
  \end{enumerate}
\end{theorem}

In Lemma~\ref{lem:copeland-unit-fpt-red}, we give a reduction from the
above problem to $\copeland$-\textsc{Shift-Bribery} with unit prices,
parameterized by the number of unit shifts (i.e., parameterized by the
value of an optimal solution).  Together with
Theorem~\ref{thm:dks-fpt-inapprox}, this reduction implies
Theorem~\ref{thm:copeland-unit-fpt} (by selecting $\delta =
\varepsilon/2$).

\begin{lemma} \label{lem:copeland-unit-fpt-red} For every
  $\alpha \in [0, 1]$ and every constant $\delta \in (0, 1)$, there
  exists a polynomial time reduction that takes in a graph $G$
  and a positive integer $k$, and outputs an instance $I'$
  of $\copeland$-\textsc{Shift-Bribery} with unit prices such that the
  following holds (where $R = \lceil 4/\delta \rceil + 2$):
  \begin{enumerate}
  \item (Completeness) If there exists a $k$-clique in $G$, then
    $\opt(I') \leqs R \cdot \binom{k}{2}$.
  \item (Soundness) If every $k$-vertex subgraph of $G$ contains fewer
    than $\delta \binom{k}{2}$ edges, then
    $\opt(I') > R (2 - 2\delta) \binom{k}{2}$.
  \end{enumerate}
\end{lemma}

Before we proceed, we note that the reduction from the previous
section does not work in the parameterized context. The reason is that
the optimum there depends on $|V_G|$, whereas here we would like the
optimum to be bounded from above by some function of $k$. Instead, we
turn to the reduction used by Bredereck et al.~\cite{BCFNN16} to
prove $\W[1]$-hardness of the problem; we only make slight
modifications to the reduction so that the analysis goes through even
for the inapproximability proof. Once again, we describe this
reduction in two steps, by first reducing to
$\copeland$-\textsc{Shift-Bribery} with $(1, \infty)$-all-or-nothing
prices and then applying Lemma~\ref{lem:copeland-red} to obtain an
instance with unit prices.

\begin{lemma} \label{lem:copeland-01-fpt} For every
  $\alpha \in [0, 1]$, there exists a polynomial time reduction that
  takes in a graph $G$ and a positive integer $k$ and outputs an
  instance $I$ of $\copeland$-\textsc{Shift-Bribery} with
  $(1, \infty)$-all-or-nothing prices of width $2$, such that the
  following holds for every $0 < \delta \leqs 1$:
  \begin{enumerate}
  \item (Completeness) If there exists a $k$-clique in $G$, then
    $\opt(I) \leqs \binom{k}{2}$.
  \item (Soundness) If every $k$-vertex subgraph of $G$ contains fewer
    than $\delta \binom{k}{2}$ edges, then
    $\opt(I) > (2 - \delta) \binom{k}{2}$.
  \end{enumerate}
\end{lemma}

\begin{proof}
  Our input consists of graph $G = (V_G, E_G)$ and positive integer
  $k$.  We create an instance $I = (E,p,\psi)$ of
  $\copeland$-\textsc{Shift-Bribery}, where
  $E = (C,V,\{\succ_v\}_{v \in V})$ is an election, $p$ is the
  preferred candidate, and $\psi = (\psi_v)_{v\in V}$ is a collection
  of price functions. We let $C = V_G \cup D \cup \{p, d\}$, where $D$
  consists of additional $|V_G|+5$ dummy candidates and $d$ is a
  candidate who will be a winner in $E$ (prior to shifting $p$).  We
  form the voter collection, together with their preference orders and
  price functions, as follows:
  \begin{enumerate}
  \item For each edge $e = \{u_1, u_2\} \in E_G$, we create two voters,
    $v_e$ and $v'_e$, such that the preference order of $v_e$ is:
    \begin{align*}
      u_1 \succ_{v_e} u_2 \succ_{v_e} p \succ_{v_e} \left<C \setminus \{p, u_1, u_2\}\right>,
    \end{align*}
    and the preference order of $v'_e$ is the reverse of that of
    $v_e$. The cost of the two possible shifts for $v_e$ are ones,
    i.e., $\psi_{v_e}(1) = \psi_{v_e}(2) = 1$, whereas the cost of all
    possible shifts for $v_e'$ are infinity, i.e.,
    $\psi_{v'_e}(1) = \cdots = \psi_{v'_e}(|C| - 3) = \infty$.
  \item We invoke Proposition~\ref{prop:dummy} with
    $A = V_G, B = D, b = k - 2$ and $a = k + 1$ to create polynomially
    many additional 
    voters so that $p$ loses pairwise elections against every
    candidate in $V_G$ by $2k - 3$ votes,
    $\scr_E(p) = |D| - k, \scr_E(d) = |D|$ and, for all
    $c \in V_G \cup D$, $\scr_E(c) \leqs (|V_G| + |D| + 3)/2 < |D|$.
    Note that we can use Proposition~\ref{prop:dummy} to obtain these scores
    because for each edge $e$, the preference orders of $v_e$ and
    $v'_e$ cancel each other out.  For all these voters, let the price of all
    possible shifts be infinity.
  \end{enumerate}

  (Completeness) Suppose that there exists a set $U \subseteq V_G$ of
  $k$ vertices that induces a $k$-clique. Consider a shift action
  where for each edge $e \in E_G[U]$, we shift $p$ to the first
  position in the preference order of $v_e$. Clearly, this solution is
  of cost $\binom{k}{2}$. Moreover, for each vertex $u \in U$, the
  shifts switch the ranks of $u$ and
  $p$ in the preference orders of $(k - 1)$ voters
  corresponding to all the edges incident to $u$. Since in the
  original election $p$ was losing the pairwise election against $u$
  by $2k - 3$ voters, $p$ wins the pairwise election against $u$ after
  the shifts. As a result, after the shifts the score of $p$ is at
  least $|D| - k + |U| = |D|$, which is at least as high as the score
  of $d$ (note that our shifts did not change $d$'s score). Hence, $p$
  is a winner of the election.

  (Soundness) Suppose contrapositively
  that $\opt(I) \leqs (2 - \delta)\binom{k}{2}$. Suppose that
  $\bs = (s_v)_{v \in V}$ is a successful shift action with cost
  $\opt(I)$.
  Since we have all-or-nothing prices,
  we can assume that for every vote where $\bs$ shifts $p$, it shifts
  him or her to the top position in the vote.
  Note that no shift action of finite cost can change the
  score of $d$. As a result, $p$ must end up having score at least
  $|D|$. Observe also that no shift action of finite cost affects the
  relative order of $p$ and any candidate from $D$. This implies that,
  after applying $\bs$, $p$ must win pairwise elections against at
  least $k$ candidates in $V_G$; let $U \subseteq V_G$ denote the set
  of $k$ candidates that $p$ ends up wining pairwise elections
  against. (If $p$ wins against more than $k$ such candidates in $V_G$
  then pick $k$ of them arbitrarily.)

  Now, let $E^* \subseteq E_G$ denote the set of all edges $e \in E_G$
  such that at least one unit shift is applied to $v_e$ (i.e.,
  $s_{v_e} > 0$). Note that
  $|E^*| = \opt(I) \leqs (2 - \delta)\binom{k}{2}$. Now, let us
  consider each vertex candidate $v \in U$. Since $p$ loses pairwise
  election against $u$ by $2k - 3$ votes in $E$, but ends up winning
  the pairwise election after applying $\bs$, it must be that applying
  $\bs$ puts $p$ ahead of $u$ in at least $k - 1$ preference
  orders. More formally, this means that:
  \begin{align*}
    k - 1 \leqs \sum_{e \in E^*} \ind[u \in e].
  \end{align*}
  Summing the above inequality over all $u \in U$ gives:
  \begin{align*}
    2\binom{k}{2} &\leqs \sum_{u \in U} \sum_{e \in E^*} \ind[u \in e] 
                  = \sum_{e \in E^*} |e \cap U| \\
                  &\leqs \bigg(\sum_{e \in (E^* \cap E_G[U])} 2\bigg) +  \bigg(\sum_{e \in (E^* \setminus E_G[U])} 1\bigg) \\
                  &= 2|E^* \cap E_G[U]| + (|E^*| - |E^* \cap E_G[U]|) \\
                  &= |E^*| + |E^* \cap E_G[U]|
                  \leqs (2 - \delta) \binom{k}{2} + |E_G[U]|,
  \end{align*}
  which implies that $|E_G[U]| \geqs \delta \binom{k}{2}$. That is,
  the subgraph of $G$ induced by $U$ is a $k$-vertex subgraph with at
  least $\delta \binom{k}{2}$ edges, which concludes our proof.
\end{proof}

We are ready to prove Lemma~\ref{lem:copeland-unit-fpt-red}, which in
turn gives Theorem~\ref{thm:copeland-unit-fpt}.

\begin{proof}[Proof of Lemma~\ref{lem:copeland-unit-fpt-red}]
  Let graph $G = (V_G, E_G)$ and positive integer $k$ be our input.
  We first apply Lemma~\ref{lem:copeland-01-fpt} to produce an
  instance $I$ of $\copeland$-\textsc{Shift-Bribery}
  with $(1,\infty)$-all-or-nothing prices and, then, we
  invoke Lemma~\ref{lem:copeland-red} with $B
  = \lceil 4/\delta \rceil$ and $B' = B(|V_G|^4 +
  1)$ to produce an instance $I'$ with unit price functions.
  The reduction runs in polynomial time.

  (Completeness) If $G$ contains a $k$-clique, then
  Lemma~\ref{lem:copeland-01-fpt} ensures that
  $\opt(I) \leqs \binom{k}{2}$. Moreover, Lemma~\ref{lem:copeland-red}
  then ensures that
  $\opt(I') \leqs (B + 2) \opt(I) \leqs R \binom{k}{2}$, where
  $R = B+2 = \lceil 4/\delta \rceil + 2$, as desired.

  (Soundness) If every $k$-vertex subgraph of $G$ contains fewer than
  $\delta \binom{k}{2}$ edges, then Lemma~\ref{lem:copeland-01-fpt}
  implies that $\opt(I) > (2 -
  \delta)\binom{k}{2}$. Lemma~\ref{lem:copeland-red} then ensures that:
  \begin{align*} \opt(I') &\geqs \min\{B', B \cdot \opt(I)\} > B (2 -
    \delta) \binom{k}{2}\\
    &\geqs  R\left(1 - \delta/2\right)(2 - \delta)\binom{k}{2}
    \geqs R(2 - 2\delta)\binom{k}{2}.
  \end{align*}
  This completes the proof.
\end{proof}

\subsubsection{Parameterization by the Number of Affected Voters}

We now move on to the parameterization by the number of affected voters
(Theorem~\ref{thm:copeland-unit-fpt-voter}). Our result will rely on
the recent parameterized hardness of approximation result for Set Cover,
due to Chen and Lin~\cite{ChenL16}. Before we state their result, let
us briefly recall the (Minimum) Set Cover problem. An instance $(U,
\cS)$ of Set Cover consists of the universe $U = \{u_1, \dots, u_N\}$
and a collection $\cS = \{S_1, \dots, S_M\}$ of subsets of $U$. The
goal is to find a subcollection $\cS' \subseteq \cS$ of smallest size
such that $\bigcup_{S \in \cS'} S = U$. The result of Chen and Lin can
be stated as follows:

\begin{theorem}[Chen and Lin~\cite{ChenL16}] \label{thm:chen-lin} Set
  Cover is $\W[1]$-hard to approximate to within any constant factor
  when parameterized by the optimal solution size.
\end{theorem}

We note that the above result has been qualitatively strengthened by
Karthik et al.~\cite{KLM18}. Indeed, if we were interested in getting
a super-constant hardness of approximation result, or a tight
running-time lower bound, the latter would give a better
result. However, we choose to state our hardness as simply as possible
(i.e., as Theorem~\ref{thm:copeland-unit-fpt-voter}) and, hence, 
it suffices to start from the result of Chen and Lin~\cite{ChenL16}.

Similarly to our previous results, we prove
Theorem~\ref{thm:copeland-unit-fpt-voter} in two steps. First, by
reducing Set Cover to $\copeland$-\textsc{Shift-Bribery} with $(1,
\infty)$-all-or-nothing prices, as stated below, and then by using
Lemma~\ref{lem:copeland-red}.

\begin{lemma} \label{lem:inapprox-voters-1} There is a polynomial time
  reduction that takes as input an instance $\tI = (U, \cS)$ of Set
  Cover and produces an instance $I$ of
  $\copeland$-\textsc{Shift-Bribery} with $(1, \infty)$-all-or-nothing
  prices such that $\opt(I) = \opt(\tI)$.
\end{lemma}

Note that the above lemma, together with Theorem~\ref{thm:chen-lin},
immediately yields hardness in
Theorem~\ref{thm:copeland-unit-fpt-voter} for $(1,
\infty)$-all-or-nothing prices. The proof of
Lemma~\ref{lem:inapprox-voters-1} is similar to the previous proofs
except that now (some) voters correspond to the subsets and (some)
candidates correspond to the elements of the universe; this can be
seen as a modification of the $\W[2]$-hardness proof of Bredereck et
al.~\cite{BCFNN16}.

\begin{proof}[Proof of Lemma~\ref{lem:inapprox-voters-1}]
  We are given an instance $\tI = (U, \cS)$ of Set Cover where $U =
  \{u_1, \dots, u_N\}$ and $\cS = \{S_1, \dots, S_M\}$. We create an
  instance $I = (E,p,\psi)$ of $\copeland$-\textsc{Shift-Bribery},
  where $E = (C,V,\{\succ_v\}_{v \in V})$ is an election, $p$ is the
  preferred candidate, and $\psi = (\psi_v)_{v\in V}$ is a collection
  of price functions. We let $C = U \cup D \cup \{p, d\}$, where $D$
  consists of additional $N + 5$ dummy candidates and $d$ is a
  candidate who will be a winner in $E$ (prior to shifting $p$).  We
  form the voter collection, together with their preference orders and
  price functions, as follows:
  \begin{enumerate}
  \item For each subset $S \in \cS$, we create two voters,
    $v_S$ and $v'_S$, such that the preference order of $v_S$ is:
    \begin{align*}
      \left<S\right> \succ_{v_S} p \succ_{v_S} \left<C \setminus (\{p\} \cup S)\right>,
    \end{align*}
    and the preference order of $v'_S$ is the reverse of that of
    $v_S$. The cost of the all possible shifts for $v_S$ are ones,
    i.e., $\psi_{v_S}(1) = \cdots = \psi_{v_S}(|S|) = 1$,
    whereas the cost of all
    possible shifts for $v_S'$ are infinity, i.e.,
    $\psi_{v'_S}(1) = \cdots = \psi_{v'_S}(|C| - |S| - 1) = \infty$.
  \item We invoke Proposition~\ref{prop:dummy} with
    $A = U, B = D, b = 0$ and $a = N + 1$ to create polynomially
    many additional 
    voters so that $p$ loses pairwise elections against every
    candidate in $U$ by $1$ vote,
    $\scr_E(p) = |D| - N, \scr_E(d) = |D|$ and, for all
    $c \in U \cup D$, $\scr_E(c) \leqs (|U| + |D| + 3)/2 < |D|$.
    Note that we can use Proposition~\ref{prop:dummy} to obtain these scores
    because for each $S \in \cS$, the preference orders of $v_S$ and
    $v'_S$ cancel each other out. For all these voters, let the price of all
    possible shifts be infinity.
  \end{enumerate}

  Clearly the reduction runs in polynomial time. We next argue that $\opt(I) = \opt(\tI)$.

  To do so, let us first argue that $\opt(I) \leq \opt(\tI)$.
  Let $\cS' \subseteq \cS$ be such that $|\cS'| = \opt(\tI)$ and
  $\bigcup_{S \in \cS'} S = U$.
  Consider a shift action that shifts $p$ to the first position in
  the preference order of $v_S$ for all $S \in \cS'$.
  This shift action is of cost $\opt(\tI)$.
  Moreover, since $\bigcup_{S \in \cS'} S = U$, this shift action switches
  the rank of $p$ and $u$ in at least one voter for each $u \in U$.
  Since $p$ was losing the pairwise election against $u$ by one vote,
  $p$ wins the pairwise election after the shifts.
  Hence, the score of $p$ after the shifts become $|D| - N + N = |D|$,
  which is at least as high as the score of $d$
  (note that our shifts did not change $d$'s score).
  Hence, $p$ is a winner of the election after the shifts,
  which implies that $\opt(I) \leq \opt(\tI)$ as desired.

  Finally, we argue that $\opt(I) \geq \opt(\tI)$.
  Suppose that $\bs = (s_v)_{v \in V}$ is a successful shift action
  with cost $\opt(I)$. Since we have all-or-nothing prices,
  we can assume that for every vote where $\bs$ shifts $p$, it shifts
  him or her to the top position in the vote.
  Note that no shift action of finite cost can change the score of $d$.
  As a result, $p$ must end up having score at least
  $|D|$. Observe also that no shift action of finite cost affects the
  relative order of $p$ and any candidate from $D$. This implies that,
  after applying $\bs$, $p$ must win pairwise elections against all
  $n$ candidates in $U$. Let $\cS'$ denote the collection of all subsets $S$
  such that at least one unit shift is applied to $v_S$ (i.e. $s_{v_S} > 0$).
  Consider any element $u \in U$. Since $p$ loses the pairwise election
  against $u$ in the original election but wins after the shifts,
  it must be that the shifts switch the order of $p$ and $u$ in
  at least one voter. This is equivalent to: $u \in S$ for some $S \in \cS'$.
  In other words, we have $\bigcup_{S \in \cS'} S = U$,
  which implies that $\opt(\tI) \leq |\cS'| = \opt(I)$ as desired.
\end{proof}

Next we use Lemma~\ref{lem:copeland-red} with appropriate values of
$B$ and $B'$ to translate the hardness from $(1,
\infty)$-all-or-nothing prices to unit prices. The properties of the
reduction from Set Cover to $\copeland$-\textsc{Shift-Bribery} with
unit prices are summarized in the lemma below.

\begin{lemma} \label{lem:copeland-unit-voter-red} There is a
  polynomial time reduction that takes as input an instance $\tI = (U,
  \cS)$ of Set Cover and produces an instance $I'$ of
  $\copeland$-\textsc{Shift-Bribery} with unit prices and $B \in
  \mathbb{N}$ such that
\begin{enumerate}[(i)]
\item any minimum cost shift action affects at most $\opt(\tI)$ voters, and, 
\item $B \cdot \opt(\tI) \leqs \opt(I') \leqs 2B \cdot \opt(\tI)$.
\end{enumerate}
\end{lemma}

We remark that Lemma~\ref{lem:copeland-unit-voter-red} immediately implies
the hardness for unit prices stated in
Theorem~\ref{thm:copeland-unit-fpt-voter}. This is because (i) implies that
the reduction is FPT (i.e. the parameter in the new problem is no more than
the previous parameter), and (ii) implies that, for any constant $\tau$,
a $\tau$-approximation for $\opt(I')$ would give a $(2\tau)$-approximation for
$\opt(\tI)$; the latter is $\W[1]$-hard (Theorem~\ref{thm:chen-lin}).

\begin{proof}[Proof of Lemma~\ref{lem:copeland-unit-voter-red}]
  Let an instance $\tI = (U, \cS)$ of Set Cover be our input.
  We first apply Lemma~\ref{lem:inapprox-voters-1} to produce an
  instance $I = (E, \psi, p)$ of $\copeland$-\textsc{Shift-Bribery}
  with $(1,\infty)$-all-or-nothing prices with election
  $E = (C, V, \{\succ_v\}_{v \in V})$.
  We then invoke Lemma~\ref{lem:copeland-red} with $B
  = b \cdot M + 1$ and $B' = B \cdot M$ to produce
  an instance $I'$ with unit price functions.

  The reduction clearly runs in polynomial time.
  Moreover, notice that $E$ has width $b$ and $b \leqs N < |C|$.
  Since $B > b \cdot M \geqs b \cdot \opt(I)$,
  Lemma~\ref{lem:copeland-red} implies that any optimal successful shift
  affects at most $\opt(I) = \opt(\tI)$ voters as desired.

  Lemma~\ref{lem:copeland-red} also yields
  $\opt(I') \geqs \min\{B', B \cdot \opt(I)\} = B \cdot \opt(\tI)$ and
  $\opt(I') \leqs (B + b) \cdot \opt(I) \leqs (2B) \cdot \opt(\tI)$,
  which completes our proof.
\end{proof}

\section{Conclusions}
We have given the first PTAS for \textsc{Shift-Bribery} for the case
of positional scoring rules, and we have shown severe limitations
regarding approximability of $\copeland$-\textsc{Shift-Bribery}. We
have also shown more efficient versions of our algorithms for the case
of the Borda rule with unit prices.
Our PTAS improves upon the $2$-approximation algorithm of Elkind et
al.~\cite{EFS09,EF10}, but their algorithm is quite robust and was
used, e.g., for combinatorial shift bribery~\cite{BFNT16-combi} and
bribery in approval elections~\cite{FST17}. It may be possible to
apply our technique in these settings as well.

Another interesting direction is to see whether our ideas can be
applied to the \textsc{Bribery} problem, where the goal is to minimize
the number of bribed voters but we are allowed to change each bribed
voter's preference arbitrarily. On this front, Keller et
al.~\cite{KellerHH18} give a PTAS for the problem for Borda,
$t$-approval, and, more generally, any scoring rule that satisfies a
certain technical condition.  It remains open whether a PTAS exists
for all scoring rules.

\section*{Acknowledgments}

Piotr Faliszewski was supported by AGH University statutory research
grant~11.11.230.337.
Pasin Manurangsi was supported by NSF Awards CCF-1813188 and CCF-1815434.
Krzysztof Sornat was partially supported by
the Foundation for Polish Science (FNP) within the START programme,
the National Science Centre, Poland,
grant numbers 2015/17/N/ST6/03684, 2018/28/T/ST6/00366 and
the SNSF Grant 200021\_159697/1.
We would like to thank the anonymous reviewers for their helpful comments.

\bibliography{main}
\bibliographystyle{alpha}

\appendix

\section{NP-hardness of Uniform-All-or-Nothing
  Borda-\textsc{Shift-Bribery}} \label{app:np-hardness-uni}

Elkind, Faliszewski and Slinko~\cite{EFS09} showed $\NP$-hardness of
Borda-\textsc{Shift-Bribery} using $(1,\infty)$-all-or-nothing prices,
whereas Bredereck et al.~\cite{BCFNN16} adapted their proof to the
case of unit prices. In this section we show that $\NP$-hardness holds
already for uniform-all-or-nothing prices.

Our reduction is similar in spirit to the one used by Elkind et
al.~\cite{EFS09}, but to avoid using $\infty$-prices, we have to
provide some more structure.  To this end, we reduce from the \vccubic
problem, that is, from the \textsc{Vertex Cover} problem on
$3$-regular graphs; it is well known that \vccubic remains
$\NP$-hard~\cite{GJS76,FHJ98}.  Formally, in \vccubic we are given a
graph $G=(V_G,E_G)$, where each vertex has degree exactly three, and a
positive integer $k$. We ask if there is a subset $A \subseteq V_G$ of
at most $k$ vertices such that each edge is incident to at least one
vertex from $A$.

The idea of the reduction is to introduce $n_G = |V_G|$ pairs of
voters that represent the vertices of $G$, and $m_G = |E_G|$
candidates that represent the edges of $G$. We interpret bribing a
voter (and shifting the preferred candidate to the top position in his
or her preference order) as including a given vertex in the cover.
Consequently, shifting a given candidate from $E_G$ back (by one
position) is interpreted as covering the corresponding edge.
The constructed instance of uniform-all-or-nothing
Borda-\textsc{Shift-Bribery} is a \emph{yes}-instance if and only if
all the candidates $E_G$ can be shifted down by at least one position
within the given budget. There are also additional voters that are
formed in such a way that bribing them is never beneficial ($p$ is
high in their rankings), but they increase the scores of the
candidates from $E_G$ so that they have more points than $p$. There
are also additional dummy candidates that are responsible for fixing
the relative scores of our main candidates.

\begin{theorem}
  Borda-\textsc{Shift}-\textsc{Bribery} with uniform-all-or-nothing prices
  in $\NP$-hard.
\end{theorem}

\begin{proof}
  Recall that $\left< B \right>$
  denotes an arbitrary (but fixed) preference order over subset of
  candidates $B$ and $\overleftarrow{\left<B\right>}$ denotes the
  reverse of the order~$\left<B\right>$.

  Let $(G,k)$ be an instance of \vccubic, let $n_G = |V_G|$, and let
  $m_G = |E_G|$. Note that $m_G = \frac{3}{2}n_G$ (because the graph
  is $3$-regular) and $k < n_G$ (otherwise we have a trivial
  \emph{yes}-instance; we also assume that $k \geq 3$).
  We construct an instance $(E,p,k)$ of the
  decision version of the uniform-all-or-nothing
  Borda-\textsc{Shift-Bribery} as follows.  We let the candidate set
  be $C = \{p\} \cup E_G \cup D$, where $D$ is a set of $3n_G-1$ dummy
  candidates.  We write $t$ to denote the highest-ranked dummy
  candidate in the order $\left<D\right>$.  Below we describe our
  collection of voters:
  \begin{enumerate}
  \item For each $u \in V_G$, we introduce two voters, $v_u$ and $v'_u$,
    with the following preference orders ($a,b,c \in E_G$ are the edges
    incident to $u$):
    \begin{align*}
      v_u: &\; \left<D\right> \succ a \succ b \succ c \succ p \succ \left<E_G
            \setminus \{a,b,c\}\right>\\
      v'_u: &\; \overleftarrow{\left<D\right>} \succ c \succ b \succ a \succ p
             \succ \overleftarrow{\left<E_G \setminus \{a,b,c\}\right>}.
    \end{align*}
  \item Let $L = n_G+2k-5$. For each $i \in [L]$, we introduce two
    voters, $v^E_i$ and $v^{E'}_i$, with the following preference
    orders:
    \begin{align*}
      v^E_i: &\; \left<E_G\right> \succ p \succ \left<D\right>, \phantom{123456789}\\
      v^{E'}_i:&\; \overleftarrow{\left<E_G\right>} \succ p \succ \overleftarrow{\left<D\right>}.\phantom{123456789}
    \end{align*}
  \item Two voters, $v_{-2}$ and $v_{-1}$, with the following preference orders:
    \begin{align*}
      v_{-2}: &\; \left<E_G\right> \succ t \succ p \succ \left<D \setminus \{t\}\right>,\\
      v_{-1}: &\; \overleftarrow{\left<E_G\right>} \succ p \succ \overleftarrow{\left<D\right>}.
    \end{align*}
  \end{enumerate}
  We note that $|C| = \frac{9}{2}n_G$ and
  $|V| = 2n_G+n_G+2k-5+2 \leqs 5n_G-3$. We see that the reduction can
  be computed in polynomial time.

  Let us now calculate the scores of candidate $p$, some arbitrary
  candidate $e \in E_G$, some arbitrary candidate $d \in D \setminus \{t\}$, and
  candidate $t$ (recall that $t \in D$ is ranked first
  in the order $\left<D\right>$):
  \begin{align*}
    \scr_E(p) &= \underbrace{(2n_G) \cdot (m_G-3)}_{\substack{\text{$m_G-3$ points from each}\\ \text{pair of voters $v_u, v'_u$}}} + \underbrace{2L \cdot |D|}_{\substack{\text{$|D|$ points from each}\\ \text{pair of voters $v^E_i, v^{E'}_i$}}} + \underbrace{(3n_G-2) + (3n_G-1)}_{\text{points from voters $v_{-2}$ and $v_{-1}$}}\\
    &= 9n_G^2 - 32n_G + 12n_Gk - 4k + 7, \\ \\
    \scr_E(e) &= \underbrace{2 \cdot (2m_G-2)}_{\substack{\text{points from two pairs of}\\ \text{voters, $v_{u_1}, v'_{u_1}$ and $v_{u_2}, v'_{u_2}$.} \\ \text{such that $e = \{u_1,u_2\}$}}} + \underbrace{(n_G-2) \cdot (m_G-4)}_{\substack{\text{points from remaining}\\ \text{pairs of voters $v_u$, $v'_u$}}} + \underbrace{(L+1)\cdot (2|D|+m_G+1)}_{\substack{\text{points from each pair of voters} \\ \text{$v^E_i, v^{E'}_i$, and from the pair $v_{-2}, v_{-1}$}}}\\
    &= 9n_G^2 - 32n_G + 15n_Gk - 2k + 8,\\  \\
    \scr_E(d) &= \underbrace{n_G \cdot (|D|+2m_G+1)}_{\substack{\text{points from each}\\ \text{pair of voters, $v_{u}, v'_{u}$}}} + \underbrace{(L+1)\cdot (|D|-1)}_{\substack{\text{points from each pair of voters} \\ \text{$v^E_i, v^{E'}_i$, and from the pair $v_{-2}, v_{-1}$}}}\\
              &= 9n_G^2 - 14n_G + 6n_Gk - 4k + 8, \\ \\
    \scr_E(t) &= \underbrace{\scr_E(d) + 1}_{\text{due to $v_{-2}$}}.
  \end{align*}

  One can readily verify that each edge candidate $e \in E_G$ has
  higher score than each candidate from~$D$. It is also clear that,
  prior to bribing voters, $p$ is not a winner of our election. We
  claim that $p$ can become a winner by bribing at most $k$ voters if
  and only if there is a cover for $G$ of size a most $k$.

  $(\Rightarrow)$ First, we show that if $(G,k)$ is a
  \emph{yes}-instance then $(E,p,k)$ also is a \emph{yes}-instance.
  Let $A \subseteq V_G$ be a size-$k$ set of vertices such that every
  edge is incident to at least one member of $A$; let $V_A$ be the set
  of corresponding voters, i.e., $V_A = \{ v_u \mid u \in A \}$.  Let
  $E'$ be an election obtained by shifting $p$ to the top position in
  all the votes in $V_A$ (i.e., we bribe the voters in the set
  $V_A$). Then we have:
  \begin{enumerate}
  \item
    $\scr_{E'}(p) = \scr_{E}(p) + k(|D|+3) = 9n_G^2 -32n_G +15n_G k
    -2k+7$, because $p$ gains $|D|+3$ points for each bribed voter.
  \item For each $e \in E_G$, we have
    $\scr_{E'}(e) \leqs \scr_{E}(e)-1=\scr_{E'}(p)$,
    because, as $A$ is a vertex cover,
    each edge candidate loses at least one point.
  \item For all $d \in D \setminus \{t\}$,
    $\scr_{E'}(d) = \scr_{E}(d)-k \leqs \scr_{E'}(p)$, because $d$
    loses $1$ point for each bribed voter.
    \item Lastly, $\scr_{E'}(t) = \scr_{E'}(d)+1 \leqs \scr_{E'}(p)$.
  \end{enumerate}
  The inequalities in the last two items hold because we assumed that
  $n_G > k \geqs 3$.
  It shows that $p$ has at least as many points as the other
  candidates hence $p$ is a winner.

  $(\Leftarrow)$ Next, we show that if $(E,p,k)$ is a
  \emph{yes}-instance of uniform-all-or-nothing Borda-\textsc{Shift-Bribery},
  then $(G,k)$ is a \emph{yes}-instance of \vccubic.  Let
  $V_A$ be a set of $k$ voters such that if we push $p$ to be ranked
  first by each of them, then $p$ becomes a winner of the election.
  After bribing $k$ voters, the score of $p$ increases at most by
  $k(|D|+3)$ (this happens if we bribe voters of the form $v_u$,
  $v'_u$, for $u \in V_G$; all the other voters rank $p$ higher and,
  thus, bribing them gives lower score increase for $p$). Thus, after
  bribing $k$ voters the score of $p$ is at most (by $e$ we mean an
  arbitrary edge candidate; all these candidates have identical
  scores):
  \[9n_G^2 -32n +7 +15nk -2k = \scr_{E}(e)-1. \]
  If $V_A$ contained at least one voter from the second or the third group,
  the score of p would be too low.
  Hence, $V_A$ contains voters from the first group only,
  and for each $e \in E_G$,
  $V_A$ must contain a voter that ranks $p$ below $e$. However, this
  means that the voters in $V_A$ must correspond to a vertex
  cover for $G$ of size at most $k$
  ($V_A$ may contain a pair of voters $v_u$, $v'_u$).
\end{proof}

The above proof gives some flexibility in fixing the scores of $p$ and
the edge candidates. Using this flexibility, one might be able to
obtain $\NP$-hardness results for scoring rules similar to Borda
(e.g., scoring rules that give some additional advantage for being
ranked first, such as those used in Formula~1 racing).

\end{document}